\pdfoutput=1
\newif\ifFull
\Fullfalse
\ifFull
\documentclass[11pt]{article}
\usepackage[margin=1in]{geometry}
\else
\documentclass[11pt]{article}
\usepackage[margin=1in]{geometry}
\fi

\renewcommand{\emph}[1]{\textit{\textbf{#1}}}
\usepackage{times}
\usepackage{amsmath}
\usepackage{amsthm}
\usepackage{graphicx}
\usepackage[noend]{algorithmic}
\usepackage{cite}

\newtheorem{theorem}{Theorem}
\newtheorem{lemma}[theorem]{Lemma}

\renewenvironment{proof}{\noindent{\bf Proof:}}{\hspace*{\fill}\rule{6pt}{6pt}\bigskip}

\begin{document}
\title{Zig-zag Sort: A Simple Deterministic Data-Oblivious \\
Sorting Algorithm Running in $O(n\log n)$ Time}

\author{Michael T.~Goodrich \\[5pt]
Department of Computer Science \\
University of California, Irvine \\
Irvine, CA 92697 USA \\
\texttt{goodrich@acm.org}
}

\date{}

\maketitle

\begin{abstract}
We describe and analyze \emph{Zig-zag Sort}---a 
deterministic data-oblivious sorting algorithm 
running in $O(n\log n)$ time that is
arguably simpler than previously known algorithms with 
similar properties, which are based on the AKS sorting network.
Because it is data-oblivious and deterministic,
Zig-zag Sort can be implemented as a simple
$O(n\log n)$-size sorting network, thereby providing a solution
to an open problem posed by Incerpi and Sedgewick in 1985.
In addition,
Zig-zag Sort is a variant of Shellsort, and is, in fact, the first 
deterministic Shellsort variant running in $O(n\log n)$ time.
The existence of such an algorithm was posed as an open problem by Plaxton 
\textit{et al.} in 1992 and also by Sedgewick in 1996.
More relevant for today, however, is the fact that
the existence of a simple data-oblivious deterministic sorting algorithm
running in $O(n\log n)$ time simplifies the ``inner-loop''
computation in several
proposed oblivious-RAM simulation methods (which utilize AKS sorting networks), 
and this, in turn, implies simplified mechanisms for
privacy-preserving data outsourcing in several cloud computing applications.
We provide both constructive and non-constructive 
implementations of Zig-zag Sort,
based on the existence of a circuit known as an \emph{$\epsilon$-halver},
such that the constant factors in our constructive 
implementations are orders of
magnitude smaller than those for constructive variants of the 
AKS sorting network,
which are also based on the use of $\epsilon$-halvers.
\end{abstract}


\section{Introduction}
An algorithm is \emph{data-oblivious} if its sequence 
of possible memory accesses is independent of its input values.
Thus, a deterministic
algorithm is data-oblivious if it makes the same sequence of memory accesses
for all its possible inputs of a given size, $n$, with the only variations
being the outputs of atomic primitive operations that are performed.
For example, a data-oblivious sorting algorithm may make 
``black-box'' use of a \emph{compare-exchange} 
operation, which is given an ordered pair of two input values, $(x,y)$, 
and returns $(x,y)$ if $x\le y$ and returns $(y,x)$ otherwise.
A sorting algorithm that uses only compare-exchange operations
is also known as a \emph{sorting network} (e.g.,
see~\cite{baddardesigning,knuthsorting}), 
since it can be viewed as a pipelined sequence of 
compare-exchange gates performed on pairs of $n$ input wires, each of which
is initially provided with an input item.
The study of data-oblivious 
sorting networks is classic in algorithm design, including
such vintage methods as bubble sort, 
Batcher's odd-even and bitonic sorting networks~\cite{Batcher:1968}, and the 
AKS sorting network~\cite{Ajtai:1983,aks-83} and 
its variations~\cite{Chvatal92,Leighton:1984,p-isn-90,Pippenger,s-snlg-09}.
In addition,
Shellsort and all its variations (e.g., see~\cite{s-asra-96})
are data-oblivious sorting algorithms,
which trace their origins
to a classic 1959 paper by the algorithm's namesake~\cite{Shell:1959}.
More recently,
examples of randomized data-oblivious sorting algorithms running
in $O(n\log n)$ time that sort with high probability include
constructions by
Goodrich~\cite{Goodrich:2011,Goo12} and Leighton and Plaxton~\cite{lp-hsn-98}.

One frustrating feature of previous work on deterministic data-oblivious 
sorting is that all known algorithms running in $O(n\log n)$ 
time~\cite{Ajtai:1983,aks-83,Chvatal92,Leighton:1984,p-isn-90,Pippenger,s-snlg-09},
which are based on the AKS sorting network, are
arguably quite complicated, while many of the known algorithms 
running in $\omega(n\log n)$ time are conceptually simple.
For instance,
given an unsorted array, $A$, of $n$ comparable items,
the Shellsort paradigm is based on the simple 
approach of making several passes up
and/or down $A$, performing compare-exchange operations between pairs
of items stored at obliviously-defined index intervals.  
Typically, the compare-exchanges are initially between pairs 
that are far apart in $A$ and the distances between
such pairs are gradually reduced from one pass to the next 
until one is certain that $A$ is sorted.
In terms of asymptotic performance,
the best previous Shellsort variant is due to Pratt~\cite{Pratt:1972},
which runs in $\Theta(n\log^2 n)$ time and is based on 
the elegant idea of comparing pairs of items separated
by intervals that determined by a monotonic sequence
of the products of powers of $2$ and $3$ less than $n$.
There has subsequently been a considerable amount of work on the
Shellsort algorithm~\cite{Shell:1959}
since its publication over 50 years ago (e.g., see~\cite{s-asra-96}), 
but none of this previous work has led to a simple deterministic data-oblivious
sorting algorithm running in $O(n\log n)$ time.

Independent of their historical appeal, 
data-oblivious algorithms 
are having a resurgence of interest 
of late, due to their applications to
privacy-preserving cloud computing.
In such applications, a client, Alice, outsources her data to an
honest-but-curious
server, Bob, who processes read/write requests for Alice. 
In order to protect her privacy, Alice must both encrypt
her data and obfuscate any data-dependent access patterns
for her data.
Fortunately, she can achieve these two goals
through any of a number of recent results
for simulating arbitrary RAM algorithms
in a privacy-preserving manner in a cloud-computing environment
using data-oblivious sorting as an
``inner-loop'' computation 
(e.g., see~\cite{dmn-pso-11,Eppstein:2010,Goldreich:1996,gm-ppa-11,Goodrich:2012:PGD}).
A modern challenge, however,
is that these simulation results either use the AKS sorting network
for this inner loop or compromise on asymptotic performance.
Thus, there is a modern motivation for 
a simple deterministic data-oblivious sorting
algorithm running in $O(n\log n)$ time.

In this paper, we provide a simple deterministic data-oblivious
sorting algorithm running in $O(n\log n)$ time, which we call \emph{Zig-zag Sort}.
This result solves the well-known 
(but admittedly vague) open problem of designing
a ``simple'' sorting network of size $O(n\log n)$,
posed by Incerpi and Sedgewick~\cite{Incerpi1985}.
Zig-zag Sort is a variant of Shellsort, and is, in fact, 
the first deterministic Shellsort variant running in $O(n\log n)$
time, which also solves open problems of 
Sedgewick~\cite{s-asra-96} and
Plaxton {\it et al.}~\cite{pps-ilbs-92,Plaxton1997}.
Zig-zag Sort differs
from previous deterministic Shellsort variants in that the increments
used in each its passes are not fixed, but instead 
vary according to ranges that are halved in 
each of $\lceil \log n\rceil$ phases.
As it turns out,
such varying increments are actually necessary to achieve an
$O(n\log n)$ running time, since any Shellsort algorithm with fixed
increments and $O(\log n)$ phases must have a running time of 
at least $\Omega(n\log^2 n/(\log\log n)^2)$, 
and any such algorithm with monotonically decreasing increments must
run in $\Omega(n\log^2 n/\log\log n)$ time, 
according to known lower
bounds~\cite{c-lbss-93,pps-ilbs-92,Plaxton1997,Poonen1993}.

In this paper, we concentrate primarily on conceptual simplicity,
with the result that the constant factors in our analysis of
Zig-zag Sort are admittedly not small.
These constant factors are nevertheless orders of magnitude smaller
than those for constructive versions of
the AKS sorting network~\cite{Ajtai:1983,aks-83}
and its recent optimization by Seiferas~\cite{s-snlg-09},
and are on par with the best non-constructive
variants of the AKS sorting network~\cite{Chvatal92,p-isn-90,Pippenger}.
Thus,
for several oblivious-RAM simulation
methods 
(e.g., see~\cite{dmn-pso-11,Eppstein:2010,Goldreich:1996,gm-ppa-11,Goodrich:2012:PGD}),
Zig-zag Sort provides a conceptually simple
alternative to the previous $O(n\log n)$-time 
deterministic data-oblivious sorting algorithms, which are all
based on the AKS sorting network.\footnote{We should stress, however, that 
  Zig-zag Sort is
  \emph{not} a parallel algorithm, like the AKS sorting network, which
  has $O(\log n)$ depth. Even when Zig-zag Sort is implemented as 
  a parallel sorting network, it still runs in $O(n\log n)$ time.
  }
The conceptual simplicity of Zig-zag Sort is not matched by a
simplicity in proving it is correct, however.  Instead, its proof of 
correctness 
is based on a fairly intricate analysis involving the tuning of several
parameters with respect to a family of potential functions. Thus,
while the Zig-zag Sort algorithm can be described in a few lines of
pseudocode,
our proof of correctness consumes much of this paper, with most of 
the details relegated to an appendix.

\section{The Zig-zag Sort Algorithm}
The Zig-zag Sort algorithm 
is based on repeated use of a procedure known as an
\emph{$\epsilon$-halver}~\cite{Ajtai:1983,aks-83,aks-92,Manos1999},
which incidentally also forms the basis for the AKS sorting network
and its variants.
\begin{itemize}
\item
An \emph{$\epsilon$-halver} is a data-oblivious procedure
that takes a pair, $(A,B)$, of arrays 
of comparable items, with each array being of size $n$,
and performs a sequence of compare-exchanges, such that, for any
$k\le n$, at most $\epsilon k$ of the largest $k$ elements 
of $A\cup B$ will be in $A$
and at most $\epsilon k$ of the smallest $k$ elements 
of $A\cup B$ will be in $B$,
where $\epsilon \ge 0$.
\end{itemize}
In addition, there is a relaxation of this definition, which is known
as an \emph{($\epsilon,\lambda)$-halver}~\cite{aks-92}:
\begin{itemize}
\item
An \emph{$(\epsilon,\lambda)$-halver} satisfies the above definition
for being an $\epsilon$-halver for $k\le \lambda n$, where $0<\lambda<1$.
\end{itemize}
We introduce a new construct, which we call a
\emph{$(\delta,\lambda)$-attenuator},
which takes this concept further:
\begin{itemize}
\item
A \emph{$(\delta,\lambda)$-attenuator} is a data-oblivious procedure
that takes a pair, $(A,B)$, of arrays
of comparable items, with each array being of size $n$,
such that $k_1$ of the largest $k$ elements of $A\cup B$ are in $A$
and $k_2$ of the smallest $k$ elements of $A\cup B$ are in $B$,
and performs a sequence of compare-exchanges such that
at most $\delta k_1$ of the largest $k$ elements will be in $A$
and at most $\delta k_2$ of the smallest $k$ elements will be in
$B$,
with $k\le \lambda n$, $0<\lambda<1$, and $\delta\ge 0$.
\end{itemize}

We give a pseudo-code description of
Zig-zag Sort in Figure~\ref{fig:zigzag}.
The name ``Zig-zag Sort'' is derived from two places that involve
procedures that could be called ``zig-zags.''
The first is in the computations performed in the outer loops, where
we make a Shellsort-style pass up a partitioning of the input
array into subarrays (in what we call
the ``outer zig'' phase) that we follow 
with a Shellsort-style pass down the sequence of subarrays (in what we call
the ``outer zag'' phase).
The second place is
inside each such loop, where we preface the set of compare-exchanges
for each pair of
consecutive subarrays by first swapping the elements in the two
subarrays, in a step we call the ``inner zig-zag'' step.
Of course, such a swapping
inverts the ordering of the elements in these two subarrays, which
were presumably put into nearly sorted order in the previous iteration.
Nevertheless, in spite of the counter-intuitive nature of this inner
zig-zag step, we show in the analysis section below that this
step is, in fact, quite useful.

\begin{figure}[hbt!]
\textbf{Algorithm} \textsf{ZigZagSort}$(A)$ 
\begin{algorithmic}[1]
\STATE $A^{(0)}_1 \leftarrow A$
\FOR {$j\leftarrow 1$ to $k$}
\FOR[splitting step] {$i\leftarrow 1$ to $2^{j-1}$}
\STATE Partition $A^{(j-1)}_i$ into halves, defining 
subarrays, $A^{(j)}_{2i-1}$ and $A^{(j)}_{2i}$, of size $n/2^j$ each
\STATE \textsf{Reduce}$(A^{(j)}_{2i-1},\, A^{(j)}_{2i})$
\ENDFOR
\FOR[outer zig] {$i\leftarrow 1$ to $2^j-1$}
\STATE Swap the items in $A^{(j)}_i$ and $A^{(j)}_{i+1}$ \hspace{1em} \{inner zig-zag\}
\STATE \textsf{Reduce}$(A^{(j)}_i,\,A^{(j)}_{i+1})$
\ENDFOR
\FOR[outer zag] {$i\leftarrow 2^j$ downto $2$}
\STATE Swap the items in $A^{(j)}_i$ and $A^{(j)}_{i-1}$ \hspace{1em} \{inner zig-zag\}
\STATE \textsf{Reduce}$(A^{(j)}_{i-1},\,A^{(j)}_{i})$
\ENDFOR
\ENDFOR
\end{algorithmic}
\vspace*{-6pt}
\caption{\label{fig:zigzag}
\textbf{Zig-zag Sort} (where $n=2^k$).
The algorithm,
\textsf{Reduce}$(A,B)$,
is simultaneously an $\epsilon$-halver, a
$(\beta,5/6)$-halver, and 
a $(\delta,5/6)$-attenuator, for appropriate values of 
$\epsilon$, $\delta$, and $\beta$.
Assuming that \textsf{Reduce} runs in $O(n)$ time, Zig-zag Sort clearly runs
in $O(n\log n)$ time.
}
\end{figure}

We illustrate, in Figure~\ref{fig:zigzag2},
how an outer zig phase would look as a sorting network.

\begin{figure}[hbt]
\vspace*{-6pt}
\begin{center}
\includegraphics[width=4.6in, trim=0.15in 4.1in 0.5in 0.15in, clip]{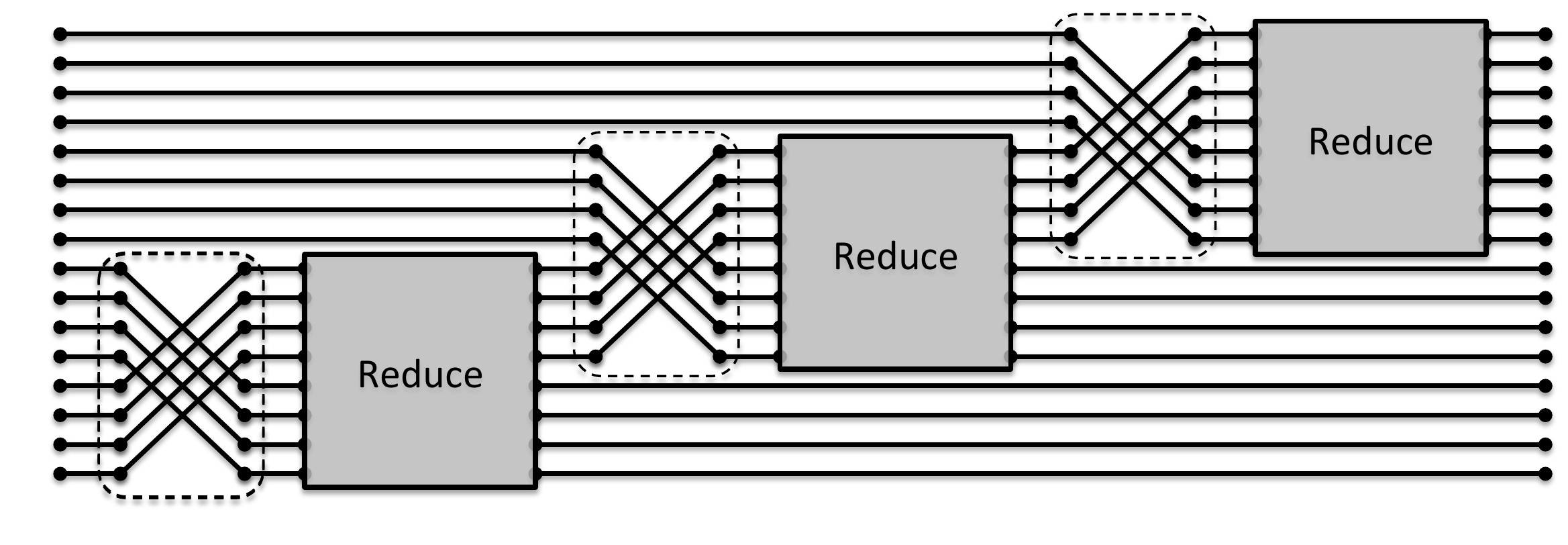}
\end{center}
\vspace*{-30pt}
\caption{An outer zig phase drawn as a sorting network, for $j=2$ and $n=16$. 
The inner zig-zag step is shown inside a dashed rounded rectangle.
Note: the inner zig-zag step could alternatively be implemented as a 
compare-exchange
of each element in the lower half with a unique element of the upper half;
we implement it as a swap, however, to reduce the total number of comparisons.}
\label{fig:zigzag2}
\end{figure}

\section{Halvers and Attenuators}
\label{sec:reduce}
In this section, we give the details for
\textsf{Reduce}, which is simultaneously an
$\epsilon$-halver, $(\beta,5/6)$-halver, and
$(\delta,5/6)$-attenuator, where the parameters, $\epsilon$,
$\beta$, and $\delta$, are functions of a single input parameter,
$\alpha>0$, determined in the
analysis section (\S\ref{sec:correct}) of this paper.
In particular, let us assume that we have a linear-time $\alpha$-halver
procedure, \textsf{Halver}, 
which operates on a pair of equal-sized arrays whose
size is a power of 2.
There are several published results for constructing such procedures 
(e.g., see~\cite{Hoory2006,Xie1998}),
so we assume the use of one of these algorithms.
The algorithm, \textsf{Reduce}, involves a call to this \textsf{Halver}
procedure and then to a recursive algorithm, \textsf{Attenuate}, which
makes additional calls to \textsf{Halver}.
See Figure~\ref{fig:reduce}.

\begin{figure}[hbt!]
\textbf{Algorithm} \textsf{Attenuate}$(A,B)$:
\begin{algorithmic}[1]
\IF {$n\le 8$}
\STATE Sort $A\cup B$ \ \ and\ \  \textbf{return}
\ENDIF
\STATE Partition $A$ into halves, defining $A^{(1)}_1$ and $A^{(1)}_2$, and
partition $B$ into halves, defining $B^{(1)}_1$ and $B^{(1)}_2$
\STATE \textsf{Halver}$(A^{(1)}_1,A^{(1)}_2)$
\STATE \textsf{Halver}$(B^{(1)}_1,B^{(1)}_2)$
\STATE \textsf{Halver}$(A^{(1)}_2,B^{(1)}_1)$
\STATE \textsf{Attenuate}$(A^{(1)}_2,B^{(1)}_1)$  
       \hspace*{2em} \{first recursive call\}
\STATE Partition $A^{(1)}_2$ into halves, defining $A^{(2)}_{1}$ and 
                                           $A^{(2)}_{2}$,
and partition $B^{(1)}_1$ into halves, defining $B^{(2)}_{1}$ and 
                                           $B^{(2)}_{2}$
\STATE \textsf{Halver}$(A^{(2)}_1,A^{(2)}_2)$

\STATE \textsf{Halver}$(B^{(2)}_1,B^{(2)}_2)$

\STATE \textsf{Halver}$(A^{(2)}_2,B^{(2)}_1)$
\STATE \textsf{Attenuate}$(A^{(2)}_{2},B^{(2)}_{1})$  \hspace*{2em} \{second 
                                           recursive call\}
\end{algorithmic}
\ifFull
\else
\medskip
\fi
\textbf{Algorithm} \textsf{Reduce}$(A,B)$: 
\begin{algorithmic}[1]
\IF {$n\le 8$}
\STATE Sort $A\cup B$ \ \ and\ \  \textbf{return}
\ENDIF
\STATE \textsf{Halver}$(A,B)$
\STATE \textsf{Attenuate}$(A,B)$
\end{algorithmic}
\vspace*{-8pt}
\caption{\label{fig:reduce}
\textbf{The \textsf{Attenuate} and \textsf{Reduce} algorithms.} 
We assume the existence of an $O(n)$-time data-oblivious procedure, 
\textsf{Halver}$(C,D)$,
which performs an $\alpha$-halver operation on two subarrays, $C$ and $D$,
each of the same power-of-2 size.
We also use a partition operation, which is just a way of viewing
a subarray, $E$, as two subarrays, $F$ and $G$, where $F$
is the first half of $E$ and $G$ is the second half of $E$.
}
\end{figure}

\pagebreak
We illustrate the data flow for the \textsf{Attenuate} algorithm
in Figure~\ref{fig:attenuate}.

\begin{figure}[hbt]
\begin{center}
\includegraphics[width=5.1in, trim=0.15in 1.0in 0.2in 0.5in, clip]{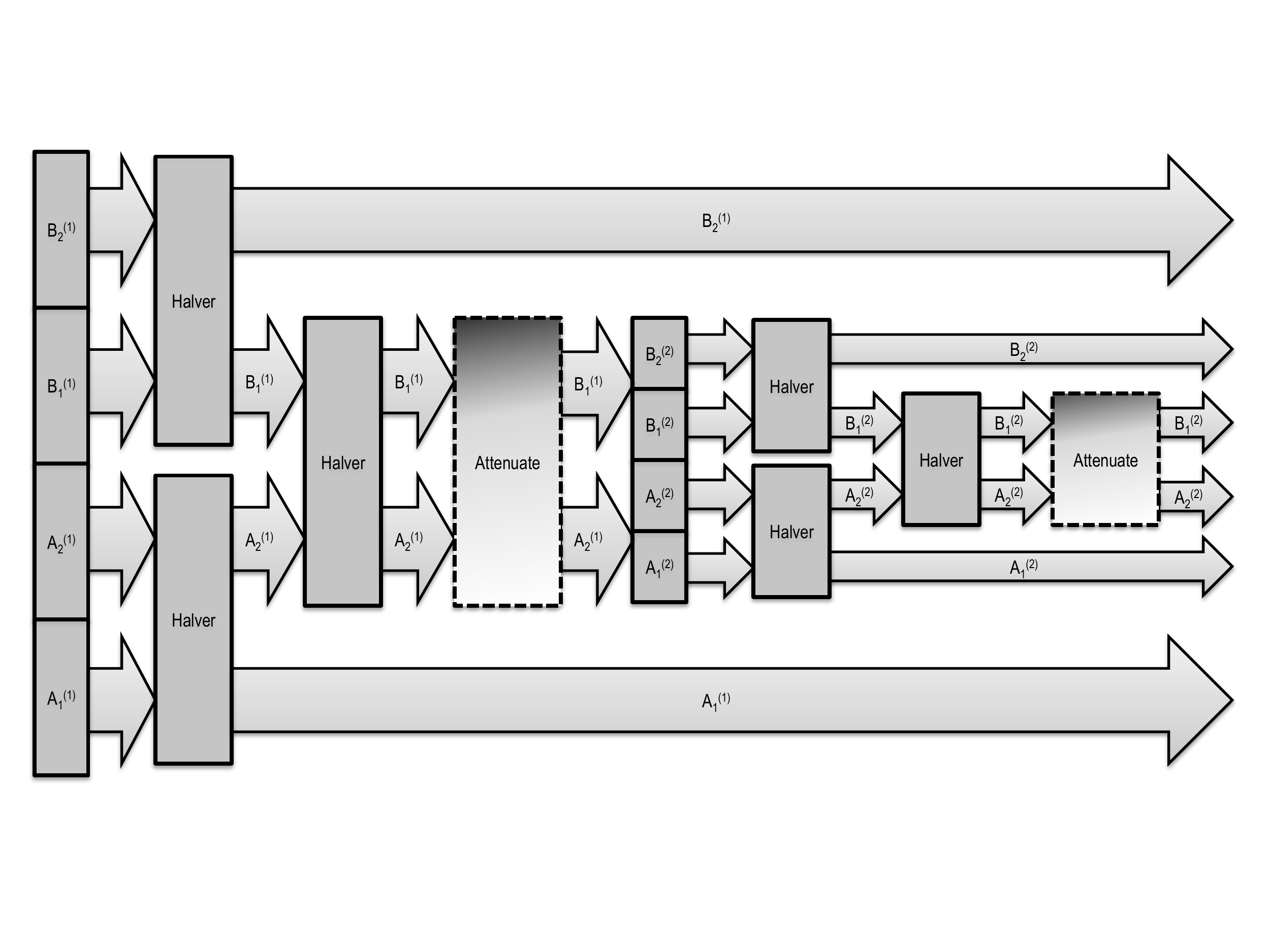}
\end{center}
\vspace*{-24pt}
\caption{Data flow in the \textsf{Attenuate} algorithm.}
\label{fig:attenuate}
\end{figure}

\section{An Analysis of Zig-Zag Sort}
\label{sec:correct}
Modulo the construction of a linear-time $\alpha$-halver procedure,
\textsf{Halver}, which we discuss in more detail
in Section~\ref{sec:constants}, 
the above discussion is a complete description of
the Zig-zag Sort algorithm.
Note, therefore, that the \textsf{Reduce} algorithm runs in $O(n)$ time,
since the running time for the general case of the recursive 
algorithm, \textsf{Attenuate}, can be characterized by the recurrence
equation,
\[
T(n) = T(n/2) + T(n/4) + bn,
\]
for some constant $b\ge1$.
In terms of the running time of Zig-zag Sort,
then,
it should be clear from the above description
that the Zig-zag Sort algorithm runs in $O(n\log n)$ time, since it performs 
$O(\log n)$ iterations, with each iteration requiring $O(n)$ time.
Proving that Zig-zag Sort is correct is less obvious, however, and
doing so consumes the bulk of the remainder of this paper.

\subsection{The 0-1 Principle}
As is common in the analysis of sorting networks (e.g.,
see~\cite{baddardesigning,knuthsorting}), 
our proof of correctness makes use of a well-known
concept known as the \emph{0-1 principle}.
\begin{theorem}[\textbf{The 0-1 Principle}~\cite{baddardesigning,knuthsorting}]
A deterministic data-oblivious (comparison-based)
sorting algorithm correctly sorts 
any input array if and only if it correctly sorts a binary array of
0's and 1's.
\end{theorem}

Thus, for the remainder of our proof of correctness, let us assume
we are operating on items whose keys are either 0 or 1. 
For instance, we use this principle in the following lemma, which we
use repeatedly in our analysis, since
there are several points when we reason about the effects
of an $\epsilon$-halver in contexts beyond its normal limits.

\begin{lemma}[\textbf{Overflow Lemma}]
\label{lem:overflow}
Suppose an $\epsilon$-halver is applied to two arrays, $A$ and $B$, of size
$n$ each, and let a parameter, $k>n$, be given.
Then at most $\epsilon n + (1-\epsilon)\cdot (k-n)$ of the $k$ largest
elements in $A\cup B$ will be in $A$
and at most $\epsilon n + (1-\epsilon)\cdot (k-n)$ of the $k$ smallest
elements in $A\cup B$ will be in $B$.
\end{lemma}
\begin{proof}
Let us focus on the bound for the $k$ largest elements, as the argument for
the $k$ smallest is similar.
By the 0-1 principle, suppose $A$ and $B$ are binary arrays, and there are $k$
1's and $2n-k$ 0's in $A\cup B$.
Since $2n-k<n$, in this case, 
after performing an $\epsilon$-halver operation, at most 
$\epsilon (2n-k)$ of the 0's will remain in $B$. That is, the number
of 1's in $B$ is at least $n-\epsilon(2n-k)$,
which implies that the number of 1's in $A$ is at most
\begin{eqnarray*}
k-(n-\epsilon (2n-k)) &=& k - n + 2\epsilon n - \epsilon k \\
&=& \epsilon n + k - n + \epsilon n - \epsilon k \\
&=& \epsilon n + (1-\epsilon)\cdot (k - n).
\end{eqnarray*}
\end{proof}

Because of the 0-1 principle, we can characterize the
distance of a subarray from being sorted by counting 
the number of 0's and 1's it contains.
Specifically, we define the \emph{dirtiness}, $D(A_i^{(j)})$, of a
subarray, $A_i^{(j)}$, 
to be the absolute value of 
the difference between the number of $1$'s currently in $A_i^{(j)}$
and the number that should be in $A_i^{(j)}$ in a final sorting of $A$.
Thus, $D(A_i^{(j)})$ counts the number of $1$'s in a subarray that
should be all $0$'s and the number of $0$'s in a subarray that should
be all $1$'s.
Any subarray of a sorted array would have a dirtiness of $0$.

\subsection{Establishing the Correctness of the \textsf{Reduce} Method}
Since the \textsf{Reduce} algorithm
comprises the main component of the Zig-zag Sort algorithm, let us begin
our detailed discussion of the correctness of Zig-zag Sort by
establishing essential properties of this algorithm.

\begin{theorem}
\label{thm:attenuate}
Given an $\alpha$-halver procedure, \textsf{Halver},
for $\alpha\le 1/6$,
which operates on arrays whose size, $n$, is a power of $2$, then
\textsf{Reduce} is a $(\delta,\,5/6)$-attenuator,
for $\delta \ge \alpha + \alpha \delta + \delta^2$.
\end{theorem}
\begin{proof}
W.l.o.g., let us analyze the number of 1's that end up in $A$; the arguments
bounding the number of 0's that end up in $B$ are similar.
Let $k\le (5/6)n$ 
denote the number of $1$'s in $A\cup B$. Also, just after the first
call to \textsf{Halver} in \textsf{Reduce}, let $k_1$ denote
the number of 1's in $A$ and let $k_2$ denote the number in $B$, so 
$k=k_1+k_2$. 
Moreover, because we preface our call to \textsf{Attenuate} in
\textsf{Reduce} with the
above-mentioned $\alpha$-halver operation, $k_1\le \alpha k\le (5\alpha/6)n$.
Also, note that if we let $k_1'$ denote 
the number of 1's in $A$ before we perform
this first $\alpha$-halver operation, 
then $k_1\le k_1'$, since any $\alpha$-halver
operation with $A$ as the first argument 
can only decrease the number of 1's in~$A$.
Note that if $n\le 8$, then 
we satisfy the claimed bound, since we reduce the number of 1's in $A$ to $0$
in this case.

Suppose, inductively, that
the recursive calls to \textsf{Attenuate} perform $(\delta,5/6)$-attenuator
operations, under the assumption that the number of 1's passed 
to the first recursive call in \textsf{Attenuate} is
at most $(5/6)n/2$ 
and that there are at most $(5/6)n/4$ passed to the second.
If $\alpha\le 1/6$, then the
results of lines~4 and~6 give us $D(A_1^{(1)})\le \alpha k_1$ and 
$D(A_2^{(1)})\le k_1$.
Thus, inductively, after the first call to \textsf{Attenuate}, we
have $D(A_2^{(1)})\le \delta k_1$. 
The results
of lines~9 and~11 give us
$D(A_1^{(2)})\le \alpha\delta k_1$ and $D(A_2^{(2)})\le \delta k_1$.
Thus, inductively, after the second call to \textsf{Attenuate}, we
have $D(A_2^{(2)})\le \delta^2 k_1$. 
Therefore, if we can show that the number of 1's passed to each call of
\textsf{Attenuate} is 5/6 of the size of the input subarrays, then we
will establish the lemma, provided that
\[
\delta \ge
\alpha + \alpha \delta + \delta^2.
\]
To bound the number of 1's passed to each recursive call to
\textsf{Attenuate}, we establish the following claim.

\textbf{Claim:} 
\textit{
The number of 1's passed to the first recursive call 
in \textsf{Attenuate} is at most $5n/12$.  
}

Since the structure of the \textsf{Attenuate} algorithm
involves the same kinds of $\alpha$-halver operations from the first recursive
call to the second, this will also imply that the number of 1's passed to the
second recursive call is at most $5n/24$, provided it holds for the first call.
To keep the constant factors reasonable, we distinguish three cases to prove
the above claim:

\begin{enumerate}
\item
Suppose $k_2\le n/2$. 
Since $k_1\le \alpha k$, in this case, $k \le n/(2-2\alpha)$,
since $k=k_1+k_2\le \alpha k + n/2$.
Here, the number of 1's passed to the recursive call is at most 
$2\alpha k + \alpha n/2$, since we start with $k_1\le \alpha k$ and
$k_2\le k$, and
\textsf{Halver}$(B^{(1)}_1,B^{(1)}_2)$ reduces the number of 1's in
$B^{(1)}_1$ in this case to be at most $\alpha k + \alpha n/2$,
by Lemma~\ref{lem:overflow}.
Thus, since,
in this case,
\[
2\alpha k + \alpha n/2 \le \alpha n/(1-\alpha) + \alpha n/2 ,
\]
the number of 1's passed to the recursive call is at most $5n/12$
if $\alpha\le 1/4.5$.

\item
Suppose $n/2 < k_2\le 2n/3$. 
Since $k_1\le \alpha k$, in this case, 
$k \le 2n/(3-3\alpha)$,
since $k=k_1+k_2\le \alpha k + 2n/3$.
Here, the number of 0's in
$B$ is $n-k_2<n/2$; hence,
the number of 0's in $B^{(1)}_2$ is at most $\alpha(n-k_2)$,
which means that the number of 1's in $B^{(1)}_2$ is at least
$n/2\,-\,\alpha(n-k_2)$, and this, in turn, implies that the
number of 1's in $B^{(1)}_1$ is at most $k_2-n/2+\alpha(n-k_2)$.
Thus, the number of 1's in the first recursive call is at most 
$k-n/2+\alpha(n-k_2)$. 
That is, it has at most $2n/(3-3\alpha) -n/2+\alpha n/2$
1's in total, which is at most $5n/12$ if $\alpha\le 1/6$.

\item
Suppose $2n/3 < k_2\le 5n/6$. 
Of course, we also know that $k\le 5n/6$ in this case.
Here, the number of 0's in $B$ is $n-k_2<n/3$; hence,
the number of 0's in $B^{(1)}_2$ is at most $\alpha(n-k_2)$,
which means that the number of 1's in $B^{(1)}_2$ is at least
$n/2\,-\,\alpha(n-k_2)$, and this, in turn, implies that the
number of 1's in $B^{(1)}_1$ is at most $k_2-n/2+\alpha(n-k_2)$.
Thus, the number of ones in the first recursive call is at most 
$k-n/2+\alpha(n-k_2)$. 
That is, it has at most $5n/6-n/2+\alpha n/3$
1's in total, which is at most $5n/12$ if $\alpha\le 1/4$.
\end{enumerate}
Thus, we have established the claim, which in turn, establishes that
\textsf{Reduce} is a $(\delta,5/6)$-attenuator, for
$\delta \ge \alpha + \alpha\delta + \delta^2$, assuming $\alpha\le 1/6$, 
since $k_1\le k_1'$.
\end{proof}

So, for example,
using a $(1/15)$-halver as the \textsf{Halver} procedure
implies that \textsf{Reduce} is
a $(1/12,5/6)$-attenuator.
%
In addition, we have the following.

\begin{theorem}
\label{thm:halver}
Given an $\alpha$-halver procedure, \textsf{Halver},
for $\alpha\le 1/6$,
which operates on arrays whose size, $n$, is a power of $2$, then
the \textsf{Reduce} algorithm is an $(\alpha\delta, 5/6)$-halver,
for $\delta\ge \alpha + \alpha\delta + \delta^2$.
\end{theorem}
\begin{proof}
Let us analyze the number of 1's that may remain in the first
array, $A$, in a call to \textsf{Reduce}$(A,B)$, as the method for
bounding the number of 0's in $B$ is similar.
After the first call to the \textsf{Halver} procedure, the number 
of 1's in $A$ is at most $\alpha k$, where $k\le (5/6)n$ is the
number of 1's in $A\cup B$.
Then, since the \textsf{Attenuate} algorithm prefaced by an
$\alpha$-halver is a $(\delta,5/6)$-attenuator,
by Theorem~\ref{thm:attenuate}, 
it will further reduce the number of 1's in $A$ to be at most
$\alpha\delta k$, where $\delta\ge \alpha+\alpha\delta+\delta^2$.
Thus, this process is an
$(\alpha\delta,5/6)$-halver.
\end{proof}

So, for example, if we construct the \textsf{Halver} procedure
to be a $(1/15)$-halver, then \textsf{Reduce} is a $(1/180,5/6)$-halver,
by Theorems~\ref{thm:attenuate} and~\ref{thm:halver}.
In addition, we have the following.

\begin{theorem}
\label{thm:better-halver}
Given an $\alpha$-halver procedure, 
\textsf{Halver},
for $\alpha\le 1/8$,
which operates on arrays whose size, $n$, is a power of $2$, then,
when prefaced by an $\alpha$-halver operation,
the \textsf{Attenuate} algorithm is an $\epsilon$-halver
for $\epsilon = \alpha^2 (\lceil \log (1/\alpha)\rceil+3)$.
\end{theorem}
\begin{proof}
Manos~\cite{Manos1999}
provides an algorithm for leveraging an $\alpha$-halver to construct
an $\epsilon$-halver,
for 
\[
\epsilon = \alpha^2 (\lceil \log (1/\alpha)\rceil+3),
\]
and every call to an $\alpha$-halver made in the
algorithm by Manos is also made in \textsf{Reduce}.
In addition, all the other calls to the $\alpha$-halver procedure made in 
\textsf{Reduce} either keep the number of 1's in $A$ unchanged or
possibly make it even smaller, since they involve compare-exchanges
between subarrays of $A$ and $B$ or they involve compare-exchanges done
after the same ones as in Manos' algorithm (and the compare-exchanges in
the \textsf{Reduce} algorithm never involve zig-zag swaps). 
Thus, the bound derived by
Manos~\cite{Manos1999} for his algorithm also applies to \textsf{Reduce}.
\end{proof}

So, for example, if we take $\alpha=1/15$, 
then \textsf{Reduce} is a $(1/32)$-halver.

\subsection{The Correctness of the Main Zig-zag Sort Algorithm}
The main theorem that establishes the correctness of the Zig-zag Sort
algorithm, given in Figure~\ref{fig:zigzag}, is the following.

\begin{theorem}
\label{thm:final}
If it is implemented using a linear-time $\alpha$-halver, \textsf{Halver}, for
$\alpha\le 1/15$,
Zig-zag Sort correctly sorts an array of $n$ comparable
items in $O(n\log n)$ time.
\end{theorem}

The details of the proof of Theorem~\ref{thm:final}
are given in the appendix, but let us nevertheless provide a sketch
of the main ideas behind the proof here.

Recall that in each iteration, $j$, of Zig-zag Sort, we divide the
array, $A$, into $2^j$ subarrays, $A_1^{(j)},\ldots,A_{2^j}^{(j)}$. 
Applying the 0-1 principle, let us
assume that $A$ stores some number, $K$, of $0$'s and $n-K$ $1$'s;
hence, in a final sorting 
of $A$, a subarray, $A_i^{(j)}$, should contain all $0$'s if 
$i<\lfloor K/2^j\rfloor$ and all $1$'s if $i>\lceil K/2^j\rceil$. 
Without loss of generality, let us assume $0<K<n$, and let us
define the index $K$ to be the \emph{cross-over} point in $A$, 
so that in a final sorting of $A$, we should have $A[K]=0$ and 
$A[K+1]=1$, by the 0-1 principle.

The overall strategy of our proof of correctness is to define a
set of potential functions upper-bounding the dirtiness 
of the subarrays in iteration $j$ while
satisfying the following constraints:
\begin{enumerate}
\item
The potential for any subarray, other than the one containing the
cross-over point, should be less than its size,
with the potential of
any subarray being a function of its distance from the cross-over
point.
\item
The potential for any subarray should be reduced in 
an iteration of Zig-zag Sort by an amount sufficient for
its two ``children'' subarrays to satisfy their dirtiness potentials
for the next iteration.
\item
The total potential of all subarrays that should contain only $0$'s
(respectively, $1$'s) should be bounded by the size of a single
subarray. 
\end{enumerate}
The first constraint ensures that $A$ will be sorted when we are
done, since the size of each subarray at the end is $1$.
The second constraint is needed in order to maintain bounds on the
potential functions from one iteration to the next.
And the third constraint is needed in order to argue that the
dirtiness in $A$ is concentrated around the cross-over point.

Defining a set of potential functions that satisfy these constraints
turned out to be
the main challenge of our correctness proof, and there are several
candidates that don't seem to work.
For example, dirtiness
bounds as an exponential function of distance 
from the cross-over 
(in terms of the number of subarrays) 
seem inappropriate,
since the capacity of the \textsf{Reduce} algorithm to move
elements is halved with each iteration, while distance from the
cross-over point is doubled, which limits our ability to
reason about how much dirtiness is removed in an outer zig or zag
phase.  
Alternatively, dirtiness bounds that are linear functions of distance
from the cross-over seem to leave too much dirtiness in ``outer''
subarrays, thereby compromising arguments that $A$ will become
sorted after $\lceil\log n\rceil$ iterations.
The particular set of potential functions that we use in our proof of
correctness, instead, can be seen as a compromise between these two
approaches.

So as to prepare for defining
the particular set of \emph{dirtiness invariants} 
for iteration $j$ that we will
show inductively holds 
after the splitting step in each iteration $j$, let us introduce a
few additional definitions.
Define the \emph{uncertainty interval} to be the set of indices
for cells in $A$
with indices
in the interval,
\[
 [K-n_j/2,\, K+1+n_j/2],
 \]
where $n_{j}=n/2^{j}$ is the size of each subarray, $A^{(j)}_i$, and
$K$ is the cross-over point in $A$.
Note that this implies that there are exactly two subarrays that
intersect the uncertainty interval in iteration $j$ of Zig-zag Sort.
In addition, for any subarray, $A^{(j)}_i$, define $d_{i,j}$ to be the 
number of iterations since this subarray has had an ancestor that was
intersecting the uncertainty interval for that level.
Also,
let $m_0$ denote the smallest index, $i$, such that $A^{(j)}_{i}$ has a cell
in the uncertainty interval and 
let $m_1$ denote the largest index, $i$, such that $A^{(j)}_{i}$ has a cell
in the uncertainty interval (we omit an implied dependence on $j$ here).
Note that these indices are defined for the sake of simplifying our
notation, since $m_1=m_0+1$.

Then,
given that the \textsf{Reduce} algorithm is 
simultaneously an $\epsilon$-halver, a
$(\beta,5/6)$-halver, and 
a $(\delta,5/6)$-attenuator, 
with the parameters, $\epsilon$, $\beta$, and $\delta$, depending on
$\alpha$, the parameter for the $\alpha$-halver, \textsf{Halver}, as
discussed in the previous section,
our potential functions and dirtiness invariants are as follows:
\begin{enumerate}
\item
After the splitting step,
for any subarray, $A^{(j)}_i$, for $i\le m_0-1$ or $i\ge m_1+1$, 
\[
D(A^{(j)}_i) \le 4^{d_{i,j}}\delta^{d_{i,j}-1} \beta n_j.
\]
\item
If the cross-over point, $K$, indexes a cell in $A^{(j)}_{m_0}$, then
$ D(A^{(j)}_{m_1}) \le n_j/6$.
\item
If the cross-over point, $K$, indexes a cell in $A^{(j)}_{m_1}$, then
$ D(A^{(j)}_{m_0}) \le n_j/6$.
\end{enumerate}

Our proof of correctness, then,
is based on arguing how the outer-zig and outer-zag phases of Zig-zag
Sort reduce the dirtiness of each subarray sufficiently to allow the
dirtiness invariants to hold for the next iteration.
Intuitively, the main idea of the arguments is to show that dirtiness
will continue to be
concentrated near the cross-over point, because the outer-zig
phase ``pushes'' $1$'s out of subarrays that should contain all $0$'s
and the outer-zag phase pushes $0$'s out of subarrays that should
contain all $1$'s.
This pushing intuition also provides the motivation behind the
inner zig-zag step, since it provides a way to ``shovel'' $1$'s
right in the outer-zig phase and shovel $0$'s left in the outer-zag
phase, where we view the subarrays of $A$ as being indexed
left-to-right.
One complication in our proof of correctness
is that this ``shoveling'' introduces some error terms in our bounds for
dirtiness during the outer-zig and outer-zag phases, 
so some care is needed to argue that these error terms
do not overwhelm our desired invariants.
The details are given in the appendix.

\section{Some Words About Constant Factors}
\label{sec:constants}
In this section, we discuss the constant factors in the running time
of Zig-zag Sort relative to the AKS sorting network and its variants.

An essential building block for Zig-zag Sort is the
existence of $\alpha$-halvers for moderately small constant values of $\alpha$,
with $\alpha\le 1/15$ being sufficient for correctness,
based on the analysis.
Suppose such an algorithm uses $cn$ compare-exchange operations,
for two subarrays whose combined size is $n$.
Then the number of compare-exchange operations performed
by the \textsf{Attenuate} algorithm is characterized by the
recurrence,
\[
T(n)=T(n/2)+T(n/4)+2.25cn,
\]
where $n$ is the total combined size of the arrays, $A$ and $B$;
hence, $T(n)=9cn$.
Thus, the running time, in terms of compare-exchange operations, 
for \textsf{Reduce}, is $10cn$, which implies that the running time
for Zig-zag Sort, in terms of compare-exchange operations, is 
at most $50cn\log n$.

An algorithm for performing a data-oblivious $\alpha$-halver
operation running in $O(n)$ time
can be built from constructions for bipartite $(\gamma,t)$-expander graphs
(e.g., see~\cite{Ajtai:1983,aks-83,Hoory2006,Xie1998}),
where such a graph,
$G=(X,Y,E)$, has the property that any 
subset $S\subset X$ of size at most $t|X|$ has at least $\gamma|S|$
neighbors in $Y$, and similarly for going from $Y$ to $X$.
Thus, if set $A=X$ and $B=Y$ and we use the edges for
compare-exchange operations, then we can build an $\alpha$-halver
from a $((1-\alpha)/\alpha,\alpha)$-expander graph with $|X|=|Y|=n$.
Also, notice that
such a bipartite graph, $G'$, can be constructed from a non-bipartite
expander graph, $G$, on $n$ vertices,
simply by making two copies, $v'$ and $v''$, in
$G'$, for every vertex $v$ in $G$, and replacing each edge $(v,w)$ in
$G$ with the edge $(v',w'')$.
Note, in addition, that $G$ and $G'$ have the same number of edges.

The original AKS sorting network~\cite{Ajtai:1983,aks-83}
is based on the use of $\epsilon$-halvers for very small
constant values of $\epsilon$ and
is estimated to have a depth of
roughly $2^{100}\log n$, meaning that the running time for simulating
it sequentially would run in roughly $2^{99}n\log n$ time in 
terms of compare-exchange operations (since 
implementing an $\epsilon$-halver sequentially halves the constant 
factor in the depth, given that every compare-exchange is between two items).
Seiferas~\cite{s-snlg-09} describes an improved scheme for building a variant
of the AKS sorting network to have $6.05\log n$ iterations, each seven
(1/402.15)-halvers deep.

By several known results 
(e.g., see~\cite{Hoory2006,Xie1998}), one can construct
an expander graph, as above, which can be used as an
$\epsilon$-halver,
using a $k$-regular graph with 
\[
k = \frac{2(1-\epsilon)(1-\epsilon+\sqrt{1-2\epsilon})}{\epsilon^2}.
\]
So, for example, if $\epsilon=1/15$, then we can construct an
$\epsilon$-halver with $cn$ edges, where $c=392$.
Using this construction results in a running time for Zig-zag
Sort of $19\,600 n\log n$, in terms of compare-exchange operations.
For the sorting network of 
Seiferas~\cite{s-snlg-09}, on the other hand, using such an
$\epsilon$-halver construction,
one can design such a 
(1/402.15)-halver to have degree $k=642\,883$; hence, the running time
of the resulting sorting algorithm would have an upper bound of
$13\,613\,047 n\log n$, in terms of compare-exchange operations.
Therefore, this constructive version of Zig-zag Sort has 
an upper bound that is three orders
of magnitude smaller than this bound for an optimized constructive
version of the AKS sorting network.

There are also non-constructive results for proving the existence of
$\epsilon$-halvers and sorting networks.
Paterson~\cite{p-isn-90} shows non-constructively that $k$-regular
$\epsilon$-halvers exist with 
\[
k = \lceil (2\log \epsilon)/\log (1-\epsilon) + 2/\epsilon - 1\rceil.
\]
So, for example, there is a
$(1/15)$-halver with $54n$ edges,
which would imply a running time of $2700 n\log n$ for Zig-zag Sort.
Using the above existence bound for the (1/402.15)-halvers used 
in the network of Seiferas~\cite{s-snlg-09}, on the other hand, 
results in a running time of $119\,025n\log n$.
Alternatively,
Paterson~\cite{p-isn-90} shows non-constructively that there exists
a sorting network of depth roughly $6100\log n$ and 
Chv{\'a}tal~\cite{Chvatal92} shows that there exists a sorting
network of depth $1830\log n$, for $n\ge 2^{78}$.
Therefore, a non-constructive version of Zig-zag Sort is competitive
with these non-constructive versions of the AKS sorting network,
while also being simpler.

\section{Conclusion}
We have given a simple deterministic data-oblivious sorting
algorithm, Zig-zag Sort, which is a variant of Shellsort running
in $O(n\log n)$ time.
This solves open problems stated by
Incerpi and Sedgewick~\cite{Incerpi1985},
Sedgewick~\cite{s-asra-96}, and
Plaxton {\it et al.}~\cite{pps-ilbs-92,Plaxton1997}.
Zig-zag Sort provides a competitive sequential 
alternative to the AKS sorting network,
particularly for applications where an explicit construction of a data-oblivious
sorting algorithm is desired, such as in applications to oblivious RAM
simulations
(e.g., see~\cite{dmn-pso-11,Eppstein:2010,Goldreich:1996,gm-ppa-11,Goodrich:2012:PGD}).

\subsection*{Acknowledgments}
This research was supported in part by
the National Science Foundation under grants 1011840, 1217322, 0916181
and 1228639, and by the Office of
Naval Research under MURI grant N00014-08-1-1015.
We would like to thank Daniel Hirschberg for several helpful comments
regarding an earlier version of this paper.

{\raggedright 
\bibliographystyle{abbrv} 
\bibliography{refs} 
}

\ifFull\else
\clearpage
\begin{appendix}
\section{The Proof of Theorem~\ref{thm:final}, Establishing the
Correctness of Zig-zag Sort}
As outlined above, our proof of Theorem~\ref{thm:final}, establishing
the correctness of Zig-zag Sort, is based on
our characterizing the dirtiness invariant for the subarrays in $A$, from
one iteration of Zig-zag Sort to the next. Let us therefore assume
we have satisfied the dirtiness invariants for a given
iteration and let us now consider how the compare-exchange operations in a
given iteration impact the dirtiness bounds for each subarray.
We establish such bounds by considering how the \textsf{Reduce}
algorithm impacts various subarrays in the outer-zig and outer-zag
steps, according to the order in which \textsf{Reduce} is called
and the distance of the different subarrays from the cross-over
point.

Recall the potential functions for our dirtiness invariants, with 
$n_j=n/2^j$:
\begin{enumerate}
\item
After the splitting step,
for any subarray, $A^{(j)}_i$, for $i\le m_0-1$ or $i\ge m_1+1$, 
\[
D(A^{(j)}_i) \le 4^{d_{i,j}}\delta^{d_{i,j}-1} \beta n_j.
\]
\item
If the cross-over point, $K$, indexes a cell in $A^{(j)}_{m_0}$, then
\[
D(A^{(j)}_{m_1}) \le n_j/6.
\]
\item
If the cross-over point, $K$, indexes a cell in $A^{(j)}_{m_1}$, then
\[
D(A^{(j)}_{m_0}) \le n_j/6.
\]
\end{enumerate}
One immediate consequence of these bounds is the following.

\begin{lemma}[Concentration of Dirtiness Lemma]
\label{lem:dirt}
The total dirtiness of all the subarrays from the 
subarray, $A^{(j)}_1$, to the subarray $A^{(j)}_{m_0-1}$,
or from the subarray, $A^{(j)}_{m_1+1}$, to the subarray $A^{(j)}_{2^j}$,
after the splitting step, is at most 
\[
\frac{8\beta n_j}{1-8\delta},
\]
provided $\delta<1/8$.
\end{lemma}
\begin{proof}
Note that, since we divide each subarray in two in each iteration
of Zig-zag Sort, there are $2^k$ subarrays, $A_i^{(j)}$, 
with depth $k=d_{i,j}$.
Thus, by the dirtiness invariants, the total
dirtiness of all the subarrays from the first
subarray, $A^{(j)}_1$, to
the subarray $A^{(j)}_{i+1}$,
after the splitting step, for $j<m_0$, is at most 
\begin{eqnarray*}
\sum_{k=1}^{j} 2^k 4^k \delta^{k-1}\beta n_j
&<& 8\beta n_j \sum_{k=0}^{\infty} (8\delta)^{k} \\
&= & \frac{8\beta n_j}{1-8\delta},
\end{eqnarray*}
provided $\delta<1/8$.
A similar argument
establishes the bound for the total dirtiness 
from the subarray, $A^{(j)}_{m_1+1}$, to the subarray $A^{(j)}_{2^j}$.
\end{proof}

\newcommand{\RD}{{\overrightarrow{D}}}
\newcommand{\LD}{{\overleftarrow{D}}}
For any subarray, $B$, of $A$, define
$\RD(B)$ to be the dirtiness of $B$ after 
the outer-zig step and $\LD(B)$ to be the dirtiness of $B$
after the outer-zag step.
Using this notation,
we begin our analysis with the following lemma, which establishes a
dirtiness bound for subarrays far to the left of the cross-over point.

\begin{lemma}[Low-side Zig Lemma]
\label{lem:low-side-zig}
Suppose the dirtiness invariants are satisfied after the splitting step 
in iteration $j$. 
Then, for $i\le m_0-2$,
after the first (outer zig) phase in iteration $j$, 
\[
\RD(A^{(j)}_i) \le \delta D(A^{(j)}_{i+1}) \le 
4^{d_{i+1,j}}\delta^{d_{i+1,j}} \beta n_j,
\]
provided $\delta\le  1/12$ and $\beta\le 1/180$.
\end{lemma}
\begin{proof}
Suppose $i\le m_0 - 2$.
Assuming that the dirtiness invariants are satisfied after the splitting step in
iteration~$j$, then, 
prior to the swaps done in the inner zig-zag step
for the subarrays $A^{(j)}_{i}$ and $A^{(j)}_{i+1}$, 
we have 
\[
D(A^{(j)}_{i+1}) \le 4^{d_{i+1,j}}\delta^{d_{i+1,j}-1} \beta n_j.
\]
In addition, by Lemma~\ref{lem:dirt},
the total dirtiness of all the subarrays from the first
subarray, $A^{(j)}_1$, to
the subarray $A^{(j)}_{i}$,
after the splitting step, is at most 
\begin{eqnarray*}
\frac{8\beta n_j}{1-8\delta} 
& \le & \frac{n_j}{6},
\end{eqnarray*}
provided $\delta\le  1/12$ and $\beta\le 1/180$.
Moreover, the cumulative way that we process the subarrays from the first
subarray to $A^{(j)}_{i}$ implies that the total amount of dirtiness brought
rightward from these subarrays to $A_i^{(j)}$ is at most the above value.
Therefore, 
after the swap of $A^{(j)}_{i}$ and $A^{(j)}_{i+1}$ in the inner zig-zag step,
the $(\delta,5/6)$-attenuator,
\textsf{Reduce},
will be effective to reduce the dirtiness for $A^{(j)}_i$
so that
\[
\RD(A^{(j)}_i) \le \delta D(A^{(j)}_{i+1}) \le 
4^{d_{i+1,j}}\delta^{d_{i+1,j}} \beta n_j.
\]
\end{proof}

As discussed at a high level earlier, 
the above proof provides a motivation for the inner zig-zag step, 
which might at first seem counter-intuitive, since it swaps many
pairs of items that are likely to be in the correct order already.
The reason we perform the inner zig-zag step, though, is that, as we
reasoned in the above proof, it
provides a way to ``shovel'' relatively large amounts of dirtiness, 
while reducing the dirtiness of all the subarrays along the way, starting
with the first subarray in $A$.
In addition to the above Low-side Zig Lemma, 
we have the following for a subarray
close to the uncertainty interval.

\begin{lemma}[Left-neighbor Zig Lemma]
\label{lem:left-neighbor-zig}
Suppose the dirtiness invariants are satisfied after the splitting step 
in iteration $j$.
If $i=m_0-1$,
then, after the first (outer zig) phase in iteration $j$, 
\[
\RD(A^{(j)}_i) \le \beta n_j,
\]
provided $\delta\le 1/12$, $\epsilon\le 1/32$, and $\beta\le 1/180$.
\end{lemma}
\begin{proof}
Since $i+1=m_0$, we are considering in this lemma
impact of calling \textsf{Reduce} on
$A^{(j)}_i$ and $A^{(j)}_{i+1} = A^{(j)}_{m_0}$, that is,
$A^{(j)}_i$ and the left subarray that intersects the uncertainty interval.
There are two cases.

\textbf{Case 1:} The cross-over point, $K$, is in 
$A^{(j)}_{m_1}$.
In this case, by the dirtiness invariant, and Lemma~\ref{lem:dirt},
even if this outer zig step has
brought all the 1's from the left rightward, the total number of 1's
in $A^{(j)}_{i} \cup A^{(j)}_{m_0}$ at this point is at most
\begin{eqnarray*}
\frac{n_j}{6} + \frac{8\beta n_j}{1-8\delta} 
& \le & \frac{n_j}{3},
\end{eqnarray*}
provided $\delta\le 1/12$ and $\beta\le 1/180$.
The above dirtiness is 
clearly less than $5n_j/6$ in this case.
Thus, the \textsf{Reduce} algorithm, which is
an $(\beta,5/6)$-halver, is
effective in this case to give us
\[
\RD(A^{(j)}_i) \le \beta n_j.
\]

\textbf{Case 2:} The cross-over point, $K$, is in $A^{(j)}_{m_0}$.
Suppose we were to sort the items currently in
$A^{(j)}_{i} \cup A^{(j)}_{m_0}$.
Then, restricted to the current state of these two subarrays,
we would get current cross-over point, $K'$,
which could be to the left, right, 
or possibly equal to the real one, $K$.
Note that each 0 that is currently to the
right of $A^{(j)}_{m_0}$
implies there must be a corresponding $1$ currently placed somewhere
from $A^{(j)}_{1}$ to $A^{(j)}_{m_0}$, possibly even 
in $A^{(j)}_{i} \cup A^{(j)}_{m_0}$.
By the dirtiness invariants for iteration~$j$,
the number of such additional 1's is bounded by 
\begin{eqnarray*}
\frac{n_j}{6} + \frac{8\beta n_j}{1-8\delta} 
& \le & \frac{n_j}{3},
\end{eqnarray*}
provided $\delta\le 1/12$ and $\beta\le 1/180$.
Thus, since there the number of 1's in $A^{(j)}_{m_0}$ 
is supposed to be at most $n_j/2$,
based on the location of the cross-over point for
this case 
(if the cross-over was closer than $n_j/2$ to the $i$th subarray, then
$i$ would be $m_0$),
the total number of 1's currently in
$A^{(j)}_{i} \cup A^{(j)}_{m_0}$
is at most $n_j/3 + n_j/2 = 5n_j/6$.
Therefore, the \textsf{Reduce} algorithm,
which is a
$(\beta,5/6)$-halver, is
effective in this case to give us
\[
\RD(A^{(j)}_i) \le \beta n_j.
\]
\end{proof}

There is also the following.

\begin{lemma}[Straddling Zig Lemma]
\label{lem:straddling-zig}
Suppose the dirtiness invariants are satisfied after the splitting step 
in iteration $j$.
Provided $\delta\le 1/12$, $\epsilon\le 1/32$, and $\beta\le 1/180$,
then after the step in first (outer zig) phase in iteration $j$, 
comparing $A^{(j)}_{m_0}$ and $A^{(j)}_{m_1}$,
we have the following:
\begin{enumerate}
\item
If $K$ is in $A^{(j)}_{m_0}$, then
\[
\RD(A^{(j)}_{m_1}) \le n_j/6 - \beta n_j.
\]
\item
If $K$ is in $A^{(j)}_{m_1}$, then
\[
\RD(A^{(j)}_{m_0}) \le n_j/6.
\]
\end{enumerate}
\end{lemma}
\begin{proof}
There are two cases.

\textbf{Case 1:} The cross-over point, $K$, is in $A^{(j)}_{m_0}$.
In this case, dirtiness for 
$A^{(j)}_{m_1}$
is caused by 0's in this subarray.
By a simple conservation argument,
such 0's can be matched with 1's that at this point in the algorithm remain
to the left of $A^{(j)}_{m_0}$.
Let $K_{m_0}$ be the index in $A^{(j)}_{m_0}$ for the cross-over
point, $K$.
By Lemmas~\ref{lem:low-side-zig}
and~\ref{lem:left-neighbor-zig}, and the fact that there are $2^k$
subarrays, $A_i^{(j)}$, with depth $k=d_{i,j}$,
the total number of 0's in 
$A^{(j)}_{m_0} \cup A^{(j)}_{m_1}$
is therefore bounded by
\begin{eqnarray*}
n' &=& K_{m_0} + \beta n_j +
\sum_{k=1}^{j} 2^k 4^k \delta^{k}\beta n_j \\
&<& K_{m_0} + \beta n_j \sum_{k=0}^{\infty} (8\delta)^{k} \\
&= & K_{m_0} + \frac{\beta n_j}{1-8\delta},
\end{eqnarray*}
provided $\delta<1/8$.
There are two subcases:
\begin{enumerate}
\item
$n'\le n_j$.
In this case, the \textsf{Reduce} algorithm,
which is an $\epsilon$-halver,
will be effective to reduce the
dirtiness of 
$A^{(j)}_{m_1}$ to $\epsilon n_j$, which is at most $n_j/6 - \beta n_j$ if
$\epsilon\le 1/8$ and $\beta\le 1/180$.
\item
$n'> n_j$.
In this case, 
note that, by the Overflow Lemma (Lemma~\ref{lem:overflow}),
the \textsf{Reduce} algorithm, which is an $\epsilon$-halver, 
will be effective to reduce the
number of 0's in
$A^{(j)}_{m_1}$, which is its dirtiness, to at most 
\begin{eqnarray*}
\epsilon n_j + (1-\epsilon)\cdot (n' - n_j) 
&=& \epsilon n_j + (1-\epsilon) \cdot \left(
K_{m_0} + \frac{\beta n_j}{1-8\delta} - n_j \right) \\
&\le&
\epsilon n_j + (1-\epsilon) \cdot \left(
\frac{\beta n_j}{1-8\delta} \right) \\
&\le&
\frac{n_j}{6} - \beta n_j,
\end{eqnarray*}
provided $\delta\le 1/12$, $\epsilon\le 1/32$, and $\beta\le 1/180$.
\end{enumerate}

\textbf{Case 2:} The cross-over point, $K$, is in $A^{(j)}_{m_1}$.
Let $K_{m_1}$ denote the index in $A_{m_1}^{(j)}$ of the cross-over
point, $K$.
In this case, dirtiness for 
$A^{(j)}_{m_0}$
is determined by 1's in this subarray.
By a simple conservation argument,
such 1's can come from 0's that at this point in the algorithm remain
to the right of $A^{(j)}_{m_1}$.
Thus, 
since there are suppose to be $(n_j-K_{m_1})$ 1's in 
$A^{(j)}_{m_0} \cup A^{(j)}_{m_1}$,
the number of 1's in these two
subarrays is bounded by
\[
n' = (n_j-K_{m_1}) + \frac{8\beta n_j}{1-8\delta},
\]
by Lemma~\ref{lem:dirt}.
There are two subcases:
\begin{enumerate}
\item
$n'\le n_j$.
In this case, the \textsf{Reduce} algorithm,
which is an $\epsilon$-halver,
will be effective to reduce the
dirtiness of 
$A^{(j)}_{m_0}$ to $\epsilon n_j$, 
which is at most $n_j/6$ if $\epsilon\le 1/6$.
\item
$n'> n_j$.
In this case, 
note that, by the Overflow Lemma (Lemma~\ref{lem:overflow}),
the \textsf{Reduce} algorithm, which is an $\epsilon$-halver, 
will be effective to reduce the
number of 1's in
$A^{(j)}_{m_0}$, which is its dirtiness, to at most 
\begin{eqnarray*}
\epsilon n_j + (1-\epsilon)\cdot (n'-n_j) 
&=&
\epsilon n_j + 
(1-\epsilon) \cdot \left(
n_j-K_{m_1} + \frac{8\beta n_j}{1-8\delta} - n_j\right) \\
&\le&
\epsilon n_j + (1-\epsilon)\cdot \left( 
\frac{8\beta n_j}{1-8\delta} \right) \\
&\le& \frac{n_j}{6},
\end{eqnarray*}
provided $\delta\le 1/12$, $\epsilon\le 1/32$, and $\beta\le 1/180$.
\end{enumerate}
\end{proof}

In addition, we have the following.

\begin{lemma}[Right-neighbor Zig Lemma]
\label{lem:right-neighbor-zig}
Suppose the dirtiness invariant is satisfied after the splitting step 
in iteration $j$.
If $i=m_1+1$,
then, after the step comparing
subarray $A^{(j)}_{m_1}$ and
subarray $A^{(j)}_{i}$ in the
first (outer zig) phase in iteration $j$, 
\[
\RD(A^{(j)}_i) \le \beta n_j,
\]
provided $\delta\le 1/12$, $\epsilon\le 1/32$, and $\beta\le 1/180$.
Also, if the cross-over is in $A^{(j)}_{m_0}$, then
$\RD(A^{(j)}_{m_1})\le n_j/6$.
\end{lemma}
\begin{proof}
Since $i-1=m_1$, we are considering in this lemma
the result of calling \textsf{Reduce} on 
$A^{(j)}_{i-1} = A^{(j)}_{m_1}$ and $A^{(j)}_i$.
There are two cases.

\textbf{Case 1:} The cross-over point, $K$, is in 
$A^{(j)}_{m_0}$.
In this case, by the dirtiness invariant
for $A_i^{(j)}$
and the previous lemma, the total number of 0's
in $A^{(j)}_{i} \cup A^{(j)}_{m_1}$ at this point is at most
$n_j/6$; hence, $\RD(A^{(j)}_{m_1})\le n_j/6$ after this comparison.
In addition, 
in this case,
the $(\beta,5/6)$-halver algorithm, \textsf{Reduce}, is
effective to give us
\[
\RD(A^{(j)}_i) \le \beta n_j.
\]

\textbf{Case 2:} The cross-over point, $K$, is in $A^{(j)}_{m_1}$.
Suppose we were to sort the items currently in
$A^{(j)}_{m_1} \cup A^{(j)}_{i}$.
Then, restricted to these two subarrays,
we would get a cross-over point, $K'$,
which could be to the left, right, 
or possibly equal to the real one, $K$.
Note that each 1 that is currently to the
left of $A^{(j)}_{m_1}$
implies there must be a corresponding $0$ currently placed somewhere
from $A^{(j)}_{m_1}$ to $A^{(j)}_{2^j}$,
possibly even 
in 
$A^{(j)}_{m_1} \cup A^{(j)}_{i}$.
The bad scenario with respect to dirtiness for
$A^{(j)}_{i}$ is when 0's are moved into
$A^{(j)}_{m_1} \cup A^{(j)}_{i}$.

By Lemmas~\ref{lem:low-side-zig},
\ref{lem:left-neighbor-zig},
and \ref{lem:straddling-zig},
and a counting argument similar to that made in the proof of 
Lemma~\ref{lem:straddling-zig},
the number of such additional 0's is bounded by 
\[
\frac{\beta n_j}{1-8\delta} + \frac{n_j}{6}.
\]
Thus, since the number of 0's 
in $A^{(j)}_{m_1}$, based on the location of the cross-over point,
is supposed to be at most $n_j/2$,
the total number of 0's currently in
$A^{(j)}_{i} \cup A^{(j)}_{m_1}$
is at most
\[
\frac{\beta n_j}{1-8\delta} + \frac{2n_j}{3}
\le 5n_j/6,
\]
provided $\delta\le 1/12$, $\epsilon\le 1/32$, and $\beta\le 1/180$.
Therefore, the $(\beta,5/6)$-halver algorithm, \textsf{Reduce}, is
effective in this case to give us
\[
\RD(A^{(j)}_i) \le \beta n_j.
\]
\end{proof}

Note that the above lemma covers the case just before we do the next 
inner zig-zag step involving
the subarrays $A^{(j)}_{i}$ 
and $A^{(j)}_{i+1}$.
For bounding the dirtiness after this inner zig-zag step we have the following.

\begin{lemma}[High-side Zig Lemma]
\label{lem:high-side-zig}
Suppose the dirtiness invariant is satisfied after the splitting step 
in iteration $j$. 
Provided $\delta\le 1/12$, $\epsilon\le 1/32$, and $\beta\le 1/180$,
then, for $m_1+1 \le i < 2^j$,
after the first (outer zig) phase in iteration $j$, 
\[
\RD(A^{(j)}_i) \le D(A^{(j)}_{i+1}) + \delta^{i-m_1-1}\beta n_j.
\]
\end{lemma}
\begin{proof}
The proof is by induction, using Lemma~\ref{lem:right-neighbor-zig}
as the base case.
We assume inductively that before we 
do the swaps for the inner zig-zag step,
$ D(A^{(j)}_i) \le \delta^{i-m_1-1} \beta n_j$.
So after we do the swapping for the inner zig-zag step and
an $(\delta,5/6)$-attenuator algorithm, \textsf{Reduce}, we have
$\RD(A^{(j)}_i) \le D(A^{(j)}_{i+1})
                 + \delta^{i-m_1-1} \beta n_j$
and
$\RD(A^{(j)}_{i+1}) \le \delta^{i-m_1}\beta n_j$.
\end{proof}

In addition, by the induction from the proof of Lemma~\ref{lem:high-side-zig},
$ \RD(A^{(j)}_{2^j}) \le \delta^{2^j-m_1-1} \beta n_j$, after we complete
the outer zig phase.
So let us consider the changes caused by the outer
zag phase.

\begin{lemma}[High-side Zag Lemma]
\label{lem:high-side-zag}
Suppose the dirtiness invariant is satisfied after the splitting step 
in iteration $j$. 
Then, for $m_1+2 \le i \le 2^j$,
after the second (outer zag) phase in iteration $j$, 
\[
\LD(A^{(j)}_i) \le 
\delta \RD(A^{(j)}_{i-1}) \le 
\delta D(A^{(j)}_{i}) + \delta^{i-m_1-1} \beta n_j
\le 4^{d_{i,j}}\delta^{d_{i,j}}\beta n_j + \delta^{i-m_1-1} \beta n_j,
\]
provided $\delta\le 1/12$, $\epsilon\le 1/32$, and $\beta\le 1/180$.
\end{lemma}
\begin{proof}
By Lemmas~\ref{lem:dirt} and~\ref{lem:high-side-zig},
and a simple induction argument,
just before we do the swaps for the inner zig-zag step,
\[
\RD(A^{(j)}_{i-1}) \le D(A^{(j)}_i)
                 + \delta^{i-m_1-2} \beta n_j
\]
and
\[
\RD(A^{(j)}_{i}) \le \frac{8\beta n_j}{1-8\delta}.
\]
We then do the inner zig-zag swaps and,
provided $\delta\le 1/12$, $\epsilon\le 1/32$, and $\beta\le 1/180$,
we have a small enough 
dirtiness to apply the $(\delta,5/6)$-attenuator, \textsf{Reduce},
effectively, which completes the proof.
\end{proof}

In addition, we have the following.

\begin{lemma}[Right-neighbor Zag Lemma]
\label{lem:right-neighbor-zag}
Suppose the dirtiness invariant is satisfied after the splitting step 
in iteration $j$.
If $i=m_1+1$,
then, after the second (outer zag) phase in iteration $j$, 
\[
\LD(A^{(j)}_i) \le \beta n_j.
\]
provided $\delta\le 1/12$, $\epsilon\le 1/32$, and $\beta\le 1/180$.
\end{lemma}
\begin{proof}
Since $i-1=m_1$, we are considering in this lemma
the result of calling \textsf{Reduce} on
$A^{(j)}_i$ and 
$A^{(j)}_{i-1} = A^{(j)}_{m_1}$.
There are two cases.

\textbf{Case 1:} The cross-over point, $K$, is in 
$A^{(j)}_{m_0}$.
In this case, by the previous lemmas, bounding the number of
0's that could have come to 
to this place from previously being in or to
the right of $A^{(j)}_{m_1}$,
the total number of 0's
in $A^{(j)}_{i} \cup A^{(j)}_{m_1}$ at this point is at most
\begin{eqnarray*}
\frac{8\beta n_j}{1-8\delta} + n_j/6 
&<& \frac{5n_j}{6},
\end{eqnarray*}
provided $\delta\le 1/12$, $\epsilon\le 1/32$, and $\beta\le 1/180$.
Thus, the $(\beta,5/6)$-halver algorithm, \textsf{Reduce}, is
effective in this case to give us
\[
\LD(A^{(j)}_i) \le \beta n_j.
\]

\textbf{Case 2:} The cross-over point, $K$, is in $A^{(j)}_{m_1}$.
Suppose we were to sort the items currently in
$A^{(j)}_{i} \cup A^{(j)}_{m_1}$.
Then, restricted to these two subarrays,
we would get a cross-over point, $K'$,
which could be to the left, right, 
or possibly equal to the real one, $c$.
Note that each 1 that is currently in or to the
left of $A^{(j)}_{m_0}$
implies there must be a corresponding $0$ currently placed somewhere
from $A^{(j)}_{m_1}$ to $A^{(j)}_{2^j}$, possibly even 
in $A^{(j)}_{m_1} \cup A^{(j)}_{i}$.
The bad scenario with respect to dirtiness for
$A^{(j)}_{i}$ is when 0's are moved into
$A^{(j)}_{i} \cup A^{(j)}_{m_1}$.
By the previous lemmas,
the number of such additional 0's (that is, 1's currently in or to
the left of $A^{(j)}_{m_0}$) is bounded by 
\[
\frac{\beta n_j}{1-8\delta} + n_j/6.
\]
Thus, since the number of 0's that 
are supposed to be in
$A^{(j)}_{m_1}$ is at most $n_j/2$ 
based on the location of the cross-over point,
the total number of 0's currently in
$A^{(j)}_{i} \cup A^{(j)}_{m_1}$
is at most  
\[
\frac{\beta n_j}{1-8\delta} + n_j/6 + n_j/2
\le 5n_j/6,
\]
provided $\delta\le 1/12$, $\epsilon\le 1/32$, and $\beta\le 1/180$.
Therefore, the $(\beta,5/6)$-halver algorithm, \textsf{Reduce}, is
effective in this case to give us
\[
\LD(A^{(j)}_i) \le \beta n_j.
\]
\end{proof}

Next, we have the following.

\begin{lemma}[Straddling Zag Lemma]
\label{lem:straddling-zag}
Suppose the dirtiness invariant is satisfied after the splitting step 
in iteration $j$
and
$\delta\le 1/12$, $\epsilon\le 1/32$, and $\beta\le 1/180$.
Then, after the comparison of $A^{(j)}_{m_0}$ and $A^{(j)}_{m_1}$,
\begin{enumerate}
\item
If $K+n_j/4$ indexes a cell in $A^{(j)}_{m_0}$,
then after the second (outer zag) phase, 
$\LD(A^{(j)}_{m_1}) \le \beta n_j$.
\item
If $K-n_j/4$ indexes a cell in $A^{(j)}_{m_1}$,
then after the second (outer zag) phase, 
$\LD(A^{(j)}_{m_0}) \le \beta n_j$.
\item
Else, 
if $K$ is in $A^{(j)}_{m_1}$, then 
$\LD(A^{(j)}_{m_0}) \le n_j/12-\beta n_j$, 
and 
if $K$ is in $A^{(j)}_{m_0}$, then 
$\LD(A^{(j)}_{m_1}) \le n_j/12 $.
\end{enumerate}
\end{lemma}
\begin{proof}
Let us consider each case.
\begin{enumerate}
\item
$K+n_j/4$ indexes a cell in $A^{(j)}_{m_0}$.
Let $K_{m_0}$ denote the index of $K$ in $A^{(j)}_{m_0}$;
hence, $K_{m_0} < 3n_j/4$.
In this case,
the number of 0's in
$A^{(j)}_{m_0} \cup A^{(j)}_{m_1}$ is at most 
$K_{m_0}$ plus at most the number of 1's that remain left of 
$A^{(j)}_{m_0}$, which is at most
\[
K_{m_0} + \frac{\beta n_j}{1-8\delta} 
\le 5n_j/6,
\]
provided $\delta\le 1/12$, $\epsilon\le 1/32$, and $\beta\le 1/180$.
Thus, the $(\beta,5/6)$-halver, \textsf{Reduce}, is effective to
give us
$\LD(A^{(j)}_{m_1})\le \beta n_j$.
\item
$K-n_j/4$ indexes a cell in $A^{(j)}_{m_1}$.
Let $K_{m_1}$ denote the index of $K$ in $A^{(j)}_{m_1}$;
hence, $K_{m_1}\ge n_j/4$.
In this case,
the number of 1's in
$A^{(j)}_{m_0} \cup A^{(j)}_{m_1}$ is at most 
$n_j - K_{m_1}$ plus the number of 0's that remain right
of 
$A^{(j)}_{m_1}$, which is at most 
\[
n_j - K_{m_1} + \frac{\beta n_j}{1-8\delta}
+ \frac{\delta\beta n_j}{1-\delta}
\le 3n_j/4 + \frac{\beta n_j}{1-8\delta}
+ \frac{\delta\beta n_j}{1-\delta}
\le 5n_j/6,
\]
provided $\delta\le 1/12$, $\epsilon\le 1/32$, and $\beta\le 1/180$.
Thus, the $(\beta,5/6)$-halver, \textsf{Reduce}, is effective to
give us
$\LD(A^{(j)}_{m_0})\le \beta n_j$.
\item
Suppose neither of the above two conditions are met.
There are two cases.

\textbf{Case 1:} The cross-over point, $K$, is in $A^{(j)}_{m_0}$.
Let $K_{m_0}$ denote the index for $K$ in $A^{(j)}_{m_0}$.
In this case, dirtiness for 
$A^{(j)}_{m_1}$
is caused by 0's coming
into $A^{(j)}_{m_0} \cup A^{(j)}_{m_1}$.
Such 0's can come from 1's that at this point in the algorithm remain
to the left of $A^{(j)}_{m_0}$.
Thus, the total number of 0's in 
$A^{(j)}_{m_0} \cup A^{(j)}_{m_1}$
is bounded by
\[
n' = K_{m_0} + \frac{\beta n_j}{1-8\delta}.
\]
There are two subcases:
\begin{enumerate}
\item
$n'\le n_j$.
In this case, the $\epsilon$-halver algorithm,
\textsf{Reduce}, will be effective to reduce the
dirtiness of 
$A^{(j)}_{m_1}$ to $\epsilon n_j$, which is at most $n_j/12$, if
$\epsilon\le 1/12$.
\item
$n'> n_j$.
By Lemma~\ref{lem:overflow},
the $\epsilon$-halver is effective
to reduce the number of 0's in
$A^{(j)}_{m_1}$, which is its dirtiness, to be at most
\begin{eqnarray*}
\epsilon n_j + (1-\epsilon)\cdot (n'-n_j) 
&=&
\epsilon n_j + 
(1+\epsilon)\cdot (K_{m_0} + \frac{\beta n_j}{1-8\delta} - n_j) \\
&\le&
\epsilon n_j + (1-\epsilon)\cdot \left(\frac{\beta n_j}{1-8\delta}\right) \\
&\le&
\frac{n_j}{12},
\end{eqnarray*}
provided $\delta\le 1/12$, $\epsilon\le 1/32$, and $\beta\le 1/180$.
\end{enumerate}

\textbf{Case 2:} The cross-over point, $K$, is in $A^{(j)}_{m_1}$.
Let $K_{m_1}$ denote the index for $K$ in $A^{(j)}_{m_1}$.
In this case, dirtiness for 
$A^{(j)}_{m_0}$
is determined by 1's in 
$A^{(j)}_{m_0} \cup A^{(j)}_{m_1}$.
Such 1's can come from 0's that at this point in the algorithm remain
to the right of $A^{(j)}_{m_1}$.
Thus, 
since there are suppose to be $(n_j-K_{m_1})$
1's in these two subarrays,
the total number of 1's in 
$A^{(j)}_{m_0} \cup A^{(j)}_{m_1}$
is bounded by
\[
n' = (n_j-K_{m_1}) + \frac{\delta\beta n_j}{1-\delta} 
     + \frac{\beta n_j}{1-8\delta}.
\]
There are two subcases:
\begin{enumerate}
\item
$n'\le n_j$.
In this case, the $\epsilon$-halver, \textsf{Reduce},
will be effective to reduce the dirtiness of 
$A^{(j)}_{m_0}$ to $\epsilon n_j$, which is at most $n_j/12-\beta n_j$,
if $\epsilon\le 1/16$ and $\beta\le 1/180$.
\item
$n'> n_j$.
In this case, by Lemma~\ref{lem:overflow},
the $\epsilon$-halver 
will be effective to reduce the
number of 1's in
$A^{(j)}_{m_0}$, which is its dirtiness, to be
at most
\begin{eqnarray*}
\epsilon n_j + (1-\epsilon)\cdot (n'-n_j) 
&=&
\epsilon n_j + 
(1-\epsilon)\cdot \left( (n_j-K_{m_1}) + \frac{\delta\beta n_j}{1-\delta} 
      + \frac{\beta n_j}{1-8\delta} - n_j \right) \\
&\le&
\epsilon n_j  + (1-\epsilon)\cdot\left(
\frac{\delta\beta n_j}{1-\delta} 
      + \frac{\beta n_j}{1-8\delta} \right) \\
&\le& \frac{n_j}{12} - \beta n_j,
\end{eqnarray*}
provided $\delta\le 1/12$, $\epsilon\le 1/32$, and $\beta\le 1/180$.
\end{enumerate}
\end{enumerate}
\end{proof}

Next, we have the following.

\begin{lemma}[Left-neighbor Zag Lemma]
\label{lem:left-neighbor-zag}
Suppose the dirtiness invariant is satisfied after the splitting step 
in iteration $j$.
Then, after the call to the
$(\beta,5/6)$-halver, \textsf{Reduce}, comparing
$A^{(j)}_i$, for $i=m_0-1$, and
$A^{(j)}_{m_0}$ in
the second (outer zag) phase in iteration $j$, 
\[
\LD(A^{(j)}_i) \le \beta n_j,
\]
provided $\delta\le 1/12$, $\epsilon\le 1/32$, and $\beta\le 1/180$.
Also, if the cross-over point, $K$, is in $A^{(j)}_{m_1}$, then 
$\LD(A^{(j)}_{m_0})\le 2\beta n_j$, if 
$K-n_j/4$ indexes a cell in $A^{(j)}_{m_1}$, 
and $\LD(A^{(j)}_{m_0})\le n_j/12$,
if $K-n_j/4$ indexes a cell in $A^{(j)}_{m_0}$.
\end{lemma}
\begin{proof}
Since $i=m_0-1$, we are considering in this lemma
the result of calling \textsf{Reduce} on
$A^{(j)}_{i+1} = A^{(j)}_{m_0}$ and $A^{(j)}_i$.
Also, note that by previous lemmas, $\RD(A^{(j)}_i)\le \beta n_j$.
There are two cases.

\textbf{Case 1:} The cross-over point, $K$, is in 
$A^{(j)}_{m_1}$.
In this case, by the dirtiness invariant
and Lemma~\ref{lem:straddling-zig},
the total number of 1's
in $A^{(j)}_{i} \cup A^{(j)}_{m_0}$ at this point is 
either $n_j/12$ or $2\beta n_j$, depending respectively on whether $K-n_j/4$
indexes a cell in $A^{(j)}_{m_0}$ or not.
Thus, in this case, $\LD(A^{(j)}_{m_0})$ is bounded by the
appropriate such bound and
the $(\beta,5/6)$-halver algorithm, \textsf{Reduce}, is
effective to give us
\[
\LD(A^{(j)}_i) \le \beta n_j.
\]

\textbf{Case 2:} The cross-over point, $K$, is in $A^{(j)}_{m_0}$.
Suppose we were to sort the items currently in
$A^{(j)}_{m_0} \cup A^{(j)}_{i}$.
Then, restricted to these two subarrays,
we would get a cross-over point, $K'$,
which could be to the left, right, 
or possibly equal to the real one, $K$.
Let $K_{m_0}$ denote the index of $K$ in $A^{(j)}_{m_0}$.
Note that each 0 that is currently to the
right of $A^{(j)}_{m_0}$
implies there must be a corresponding $1$ currently placed somewhere
from $A^{(j)}_{1}$ to $A^{(j)}_{m_0}$,
possibly even 
in 
$A^{(j)}_{i}\cup A^{(j)}_{m_0}$.
The bad scenario with respect to dirtiness for
$A^{(j)}_{i}$ is when 1's are moved into
$A^{(j)}_{i}\cup A^{(j)}_{m_0}$.
By previous lemmas,
the number of such additional 1's is bounded by 
\[
\frac{\beta n_j}{1-8\delta} + \frac{\delta\beta n_j}{1-\delta} + n_j/12.
\]
Thus, since the number of 1's 
in $A^{(j)}_{m_0}$, based on the location of the cross-over point,
is supposed to be $n_j-K_{m_0}\le n_j/2$,
the total number of 1's currently in
$A^{(j)}_{i} \cup A^{(j)}_{m_0}$
is at most
\[
\frac{\beta n_j}{1-8\delta} + \frac{\delta\beta n_j}{1-\delta} + 7n_j/12
\le 5n_j/6,
\]
provided $\delta\le 1/12$, $\epsilon\le 1/32$, and $\beta\le 1/180$.
Therefore, the $(\beta,5/6)$-halver algorithm, \textsf{Reduce}, is
effective in this case to give us
\[
\LD(A^{(j)}_i) \le \beta n_j.
\]
\end{proof}

Finally, we have the following.

\begin{lemma}[Low-side Zag Lemma]
\label{lem:low-side-zag}
Suppose the dirtiness invariant is satisfied after the splitting step 
in iteration $j$. 
If 
$\delta\le 1/12$, $\epsilon\le 1/32$, and $\beta\le 1/180$,
then, for $i\le m_0-1$,
after the second (outer zag) phase in iteration $j$, 
\[
\LD(A^{(j)}_{i}) \le \delta D(A^{(j)}_i) + \delta^{m_0-i-1} \beta n_j.
\]
\end{lemma}
\begin{proof}
The proof is by induction on $m_0-i$, starting with
Lemma~\ref{lem:left-neighbor-zag} as the basis of the induction.
Before doing the swapping for the inner zig-zag step for subarray $i$, 
by Lemma~\ref{lem:low-side-zig},
\[
\RD(A^{(j)}_{i-1}) \le \delta D(A^{(j)}_i)
\]
and
\[
\LD(A^{(j)}_{i}) \le \delta^{m_0-i-1} \beta n_j.
\]
Thus, after the swaps for the inner zig-zag and
the $(\delta,5/6)$-attenuator algorithm, \textsf{Reduce},
\[
\LD(A^{(j)}_{i}) \le \delta D(A^{(j)}_i)
+ \delta^{m_0-i-1} \beta n_j.
\]
and
\[
\LD(A^{(j)}_{i-1}) \le \delta^{m_0-i} \beta n_j.
\]
\end{proof}

This completes all the lemmas we need in order to calculate bounds
for $\epsilon$ and $\delta$ that will allow us to satisfy the
dirtiness invariant for iteration~$j+1$ if it is satisfied for
iteration~$j$.

\begin{lemma}
\label{lem:final}
Provided $\delta\le 1/12$, $\epsilon\le 1/32$, and $\beta\le 1/180$,
if the dirtiness invariant for iteration~$j$ is satisfied after the
splitting step for iteration~$j$,
then the dirtiness invariant for iteration~$j$ is satisfied after the
splitting step for iteration~$j+1$.
\end{lemma}
\begin{proof}
Let us consider each subarray, $A^{(j)}_i$, and its two children,
$A^{(j+1)}_{2i-1}$
and $A^{(j+1)}_{2i}$, at the point in the algorithm when we 
perform the splitting step.
Let $m'_0$ denote the index of the lowest-indexed subarray on level $j+1$ that
intersects the uncertainty interval, and
let $m'_1 \, (=m'_0+1)$ denote the index of the 
highest-indexed subarray on level $j+1$ that
intersects the uncertainty interval.
Note that we either have $m'_0$ and $m'_1$ both being children of $m_0$,
$m'_0$ and $m'_1$ both being children of $m_1$,
or $m'_0$ is a child of $m_0$ and $m'_1$ is a child of $m_1$.
That is, $m'_0=2m_0-1$, 
$m'_0=2m_0$, 
or $m'_0=2m_1-1=2m_0+1$,
and $m'_1=2m_0=2m_1-2$, 
$m'_1=2m_1-1$, 
or $m'_1=2m_1$,
\begin{enumerate}
\item
$i \le m_0-1$.
In the worst case,
based on the three possibilities for $m'_0$ and $m'_1$,
we need
\[
D(A^{(j+1)}_{2i-1}) 
\le 4^{d_{2i-1,j+1}}\delta^{d_{2i-1,j+1}-1}\beta n_j
= 4^{d_{i,j}+1}\delta^{d_{i,j}}\beta n_j
\]
and
\[
D(A^{(j+1)}_{2i}) 
\le 4^{d_{2i,j+1}}\delta^{d_{2i,j+1}-1}\beta n_j
= 4^{d_{i,j}+1}\delta^{d_{i,j}}\beta n_j.
\]
By Lemma~\ref{lem:low-side-zag},
just before the splitting step for $A^{(j)}_i$, we have
\[
\LD(A^{(j)}_{i}) \le \delta D(A^{(j)}_i) +
                \delta^{m_0-i-1}\beta n_j
		\le 4^{d_{i,j}}\delta^{d_{i,j}}\beta n_j
                 + \delta^{m_0-i-1}\beta n_j
		\le (4^{d_{i,j}}+1)\delta^{d_{i,j}}\beta n_j,
\]
and we then partition $A^{(j)}_i$, so that $n_{j+1}=n_j/2$.
Thus, for either $k=2i-1$ or $k=2i$,
we have
\begin{eqnarray*}
D(A^{(j+1)}_{k}) &\le& 
		2(4^{d_{i,j}}+1)\delta^{d_{i,j}}\beta n_{j+1} \\
		&\le & 4^{d_{i,j}+1}\delta^{d_{i,j}}\beta n_{j+1},
\end{eqnarray*}
which satisfies the dirtiness invariant for the next iteration.

\item
$i\ge m_1+1$.
In the worst case,
based on the three possibilities for $m'_0$ and $m'_1$,
we need
\[
D(A^{(j+1)}_{2i-1}) 
\le 4^{d_{2i-1,j+1}}\delta^{d_{2i-1,j+1}-1}\beta n_j
= 4^{d_{i,j}+1}\delta^{d_{i,j}}\beta n_j
\]
and
\[
D(A^{(j+1)}_{2i}) 
\le 4^{d_{2i,j+1}}\delta^{d_{2i,j+1}-1}\beta n_j
= 4^{d_{i,j}+1}\delta^{d_{i,j}}\beta n_j.
\]
By Lemma~\ref{lem:high-side-zag},
just before the splitting step for $A^{(j)}_i$, we have
\[
\LD(A^{(j)}_{i}) \le \delta D(A^{(j)}_i) +
                \delta^{i-m_1-1}\beta n_j
		\le 4^{d_{i,j}}\delta^{d_{i,j}}\beta n_j
                 + \delta^{i-m_1-1}\beta n_j
		\le (4^{d_{i,j}}+1)\delta^{d_{i,j}}\beta n_j,
\]
and we then partition $A^{(j)}_i$, so that $n_{j+1}=n_j/2$.
Thus, for either $k=2i-1$ or $k=2i$,
we have
\begin{eqnarray*}
D(A^{(j+1)}_{k}) &\le& 
		2(4^{d_{i,j}}+1)\delta^{d_{i,j}}\beta n_{j+1} \\
		&\le & 4^{d_{i,j}+1}\delta^{d_{i,j}}\beta n_{j+1},
\end{eqnarray*}
which satisfies the dirtiness invariant for the next iteration.

\item
$i = m_0$.
In this case, there are subcases.
\begin{enumerate}
\item
Suppose $K+n_j/4$ indexes a cell in $A^{(j)}_{m_0}$.
Then, by Lemma~\ref{lem:straddling-zag},
$\LD(A^{(j)}_{m_1}) \le \beta n_j$.
In this case, $m'_0$ and $m'_1$ are both children of $m_0$
and we need $A^{(j+1)}_{2i-1}$
to have dirtiness at most $n_{j+1}/6=n_j/12$.
The number of 1's in 
$A^{(j)}_i$ is bounded by $n_j/2$ (or otherwise, $i=m_1$)
plus the number
of additional 1's that may be here because of 0's that remain to the right,
which is bounded by 
\[
n' = n_j/2  + \beta n_j + \frac{\beta n_j}{1-8\delta} 
                   + \frac{\delta\beta n_j}{1-\delta},
\]
provided $\delta\le 1/12$, $\epsilon\le 1/32$, and $\beta\le 1/180$.
If $n'\le n_j/2$, then an $\epsilon$-halver operation applied after 
the split, with $\epsilon\le 1/6$,
will satisfy the dirtiness invariants for $m_0'$ and $m_1'$.
If, on the other hand, $n'>n_j/2$,
then an $\epsilon$-halver operation applied after 
the split will give us 
\begin{eqnarray*}
D(A^{(j+1)}_{2i-1}) &\le& 
                   \epsilon n_{j+1} + (1-\epsilon)\cdot (n'-n_{j+1}) \\
 &\le &          \epsilon n_{j+1} + (1-\epsilon)\cdot 
              \left( 2\beta n_{j+1} + \frac{2\beta n_{j+1}}{1-8\delta} 
                   + \frac{2\delta\beta n_{j+1}}{1-\delta} \right) \\
 &\le& n_{j+1}/6,
\end{eqnarray*}
provided $\delta\le 1/12$, $\epsilon\le 1/32$, and $\beta\le 1/180$.

\item
Suppose $K-n_j/4$ indexes a cell in $A^{(j)}_{m_1}$.
Then, by Lemma~\ref{lem:left-neighbor-zag},
$\LD(A^{(j)}_{m_0}) \le 2\beta n_j$.
In this case, we need 
$D(A^{(j+1)}_{2i}) \le 4\beta n_{j+1}$
and $D(A^{(j+1)}_{2i-1}) \le 4\beta n_{j+1}$, both of which
which follow from the above
bound.

\item
Suppose neither of the previous subcases hold.
Then we have two possibilities:
\begin{enumerate}
\item
Suppose $K$ indexes a cell in $A^{(j)}_{m_1}$. 
Then $\LD(A^{(j)}_{m_0}) \le n_j/12$,
by Lemma~\ref{lem:left-neighbor-zag}.
In this case, we need $D(A^{(j+1)}_{2i})\le n_{j+1}/6$, which follows
immediately from this bound, and we also need
$D(A^{(j+1)}_{2i-1})\le 4\beta n_{j+1}$, which follows
by our performing a \textsf{Reduce} step,
which is a $(\beta,5/6)$-halver, after we do our split.

\item
Suppose $K$ indexes a cell in $A^{(j)}_{m_0}$ (but $K+n_j/4$ indexes a cell
in $A^{(j)}_{m_1}$). Then, 
by Lemma~\ref{lem:straddling-zag},
$\LD(A^{(j)}_{m_1}) \le n_j/12$.
In this case, we 
need
$D(A^{(j+1)}_{2i-1})\le 4\beta n_{j+1}$.
Here, the dirtiness of
$A^{(j+1)}_{2i-1}$ is determined by the number of 1's 
it contains, which is bounded by the intended number of 1's, 
which itself is bounded by
$n_j/4 = n_{j+1}/2$, plus the number of 0's currently to the 
right of $A^{(j)}_{m_0}$, which, all together, is at most
\[
n_{j+1}/2 + n_{j+1}/6 
          + \frac{2\beta n_{j+1}}{1-8\delta} 
	  + \frac{2\delta\beta n_{j+1}}{1-\delta}
\le 5 n_{j+1}/6,
\]
provided $\delta\le 1/12$, $\epsilon\le 1/32$, and $\beta\le 1/180$.
Thus, the \textsf{Reduce} algorithm, which is a $(\beta,5/6)$-halver,
we perform after the split will give us
$D(A^{(j+1)}_{2i-1})\le \beta n_{j+1}$.
\end{enumerate}
\end{enumerate}
\item
$i=m_1$.
In this case, there are subcases.
\begin{enumerate}
\item
Suppose $K+n_j/4$ indexes a cell in $A^{(j)}_{m_0}$.
Then, by Lemma~\ref{lem:straddling-zag},
$\LD(A^{(j)}_{m_1}) \le \beta n_j$.
In this case, we need 
$D(A^{(j+1)}_{2i-1}) \le 4\beta n_{j+1}$, 
and $D(A^{(j+1)}_{2i}) \le 4\beta n_{j+1}$, 
which follows from the above
bound.
\item
Suppose $K-n_j/4$ indexes a cell in $A^{(j)}_{m_1}$.
Then, by Lemma~\ref{lem:left-neighbor-zag},
$\LD(A^{(j)}_{m_0}) \le 2\beta n_j$.
In this case, $m'_0$ and $m'_1$ are both children of $m_1$, and the cross-over
is in $A^{(j+1)}_{m'_0}$.
We therefore need $A^{(j+1)}_{2i}$
to have dirtiness at most $n_{j+1}/6=n_j/12$.
The number of 0's in $A^{(j)}_i$ is at most $n_j/2$ (or this
wouldn't be $m_1$), plus the number
of additional 0's that are here because of 1's to the left of $m_1$,
which is bounded by 
\[
n' = n_j/2 + \beta n_j + \frac{\beta n_j}{1-8\delta} 
                         + \frac{\delta\beta n_j}{1-\delta}.
\]
If $c'\le n_{j+1}=n_j/2$, then the $\epsilon$-halver will reduce 
the dirtiness so that
$D(A^{(j+1)}_{m'_1}) \le \epsilon n_{j+1}\le n_{j+1}/6$, if $\epsilon\le 1/6$.
If, on the other hand, $n'> n_{j+1}$, then, by Lemma~\ref{lem:overflow},
$D(A^{(j+1)}_{m'_1})$ will
be reduced to be at most
\begin{eqnarray*}
\epsilon n_{j+1} + (1-\epsilon) \cdot (n'-n_{j+1})
&\le& \epsilon n_{j+1} + (1-\epsilon)\cdot 
          \left( 2\beta n_{j+1} + \frac{2\beta n_{j+1}}{1-8\delta} 
                         + \frac{2\delta\beta n_{j+1}}{1-\delta}\right) \\
&\le & \frac{n_{j+1}}{6},
\end{eqnarray*}
provided $\delta\le 1/12$, $\epsilon\le 1/32$, and $\beta\le 1/180$.
\item
Suppose neither of the previous subcases hold.
Then we have two possibilities:
\begin{enumerate}
\item
Suppose $K$ indexes a cell in $A^{(j)}_{m_0}$. 
Then $\LD(A^{(j)}_{m_1}) \le n_j/12$,
by Lemma~\ref{lem:straddling-zag}.
In this case, we need $D(A^{(j+1)}_{2i-1})\le n_{j+1}/6$, which follows
immediately from this bound, and we also need
$D(A^{(j+1)}_{2i})\le 4\beta n_{j+1}$, which follows
by our performing a \textsf{Reduce} step after we do our split.
\item
Suppose $K$ indexes a cell in $A^{(j)}_{m_1}$. Then, 
by Lemma~\ref{lem:left-neighbor-zag},
$\LD(A^{(j)}_{m_0}) \le n_j/12$.
In this case, we 
need
$D(A^{(j+1)}_{2i})\le 4\beta n_{j+1}$.
Here, the dirtiness of
$A^{(j+1)}_{2i-1}$ is determined by the number of 0's 
it contains, which is bounded by the proper number of 0's, 
which is at most
$n_j/4 = n_{j+1}/2$, plus the number of 1's currently to the 
left of $A^{(j)}_{m_1}$, which, all together, is at most
\[
n_{j+1}/2 + n_{j+1}/6 + \frac{2\beta n_{j+1}}{1-8\delta} 
          + \frac{2\delta\beta n_{j+1}}{1-\delta} 
\le 5 n_j/6,
\]
provided $\delta\le 1/12$, $\epsilon\le 1/32$, and $\beta\le 1/180$.
Thus, the \textsf{Reduce} algorithm we perform after the split will give us
$D(A^{(j+1)}_{2i-1})\le \beta n_{j+1}$.
\end{enumerate}
\end{enumerate}
\end{enumerate}
\end{proof}

Putting everything together, we establish 
the following.

\bigskip
\noindent
\textbf{Theorem~\ref{thm:final}}:
\textit{
If it is implemented using a linear-time $\alpha$-halver, \textsf{Halver}, for
$\alpha\le 1/15$,
Zig-zag Sort correctly sorts an array of $n$ comparable
items in $O(n\log n)$ time.}

\medskip
\begin{proof}
Take $\alpha\le 1/15$, for \textsf{Halver} being an $\alpha$-halver,
so that \textsf{Reduce}
is simultaneously an $\epsilon$-halver, a
$(\beta,5/6)$-halver, and 
a $(\delta,5/6)$-attenuator, for 
$\delta\le 1/12$, $\epsilon\le 1/32$, and $\beta\le 1/180$.
Such bounds achieve the necessary constraints 
for Lemmas~\ref{lem:low-side-zig} to \ref{lem:final}, given above,
which establishes the dirtiness invariants for each iteration of
Zig-zag Sort.
The correctness follows, then, by noticing that satisfying the
dirtiness invariant after the last iteration of Zig-zag Sort implies
that the array $A$ is sorted.
\end{proof}

\end{appendix}
\fi

\end{document}